\documentclass[reprint,amsmath,amssymb,aps]{revtex4-2}

\usepackage{graphicx}
\usepackage{bm}
\usepackage{mathtools}
\usepackage{xcolor}
\usepackage{url}

\usepackage[colorlinks=true,
            linkcolor=blue,     
            citecolor=blue,     
            urlcolor=blue]{hyperref}

\usepackage{tikz}
\usetikzlibrary{arrows.meta,positioning,calc}

\newtheorem{proposition}{Proposition}

\begin{document}

\preprint{APS/123-QED}

\title{Glassy dynamics of metropolitan human mobility frozen near the gravitational equilibrium}

\author{Yixuan Y. Zheng}
\email{yixuan.z.5b0a@m.isct.ac.jp}
\affiliation{Department of Systems and Control Engineering,
Institute of Science Tokyo, Yokohama, Kanagawa 226-8502, Japan}

\author{Hideki Takayasu}
\email{hideki.takayasu@gmail.com}
\affiliation{Department of Computer Science,
Institute of Science Tokyo, Yokohama, Kanagawa 226-8502, Japan}

\author{Misako Takayasu}
\email{takayasu@comp.isct.ac.jp}
\thanks{Corresponding author}
\affiliation{Department of Systems and Control Engineering,
Institute of Science Tokyo, Yokohama, Kanagawa 226-8502, Japan}
\affiliation{Department of Computer Science,
Institute of Science Tokyo, Yokohama, Kanagawa 226-8502, Japan}

\graphicspath{{images/}}

\date{\today}

\begin{abstract}

The doubly constrained gravity model is one of the most widely used descriptions of urban commuting flows, yet whether residential relocation can actually drive a city to this equilibrium has not been tested. We address this question by connecting urban commuting to thermal equilibrium. Using mobility data from around 30 million records across six Japanese cities over one year, we introduce the Gravity-based Origin Swapping Model (GOSM), which applies Metropolis dynamics to urban commuting by treating individuals as interacting particles. Within this framework, total commuting time defines the system energy, and temperature controls how strongly mobility responds to cost savings. Consequently, urban commuting can be analyzed through an evaluable free energy and simulated similarly to physical systems of matter, with the maximum-entropy doubly constrained gravity model serving as the stationary state. We find that the residential dynamics exhibit glassy characteristics: calibration to observed data reveals signatures of kinetically constrained models, specifically aging, hysteresis, and freezing near equilibrium. These results suggest that static OD snapshots may misestimate the effect of structural intervention: substantial improvements may require changing the urban landscape itself---through transport investment that reshapes travel times or housing development that reshapes residential capacities, rather than waiting for spontaneous relocation to reach the theoretical optimum.

\end{abstract}

\maketitle

\section{Introduction}
\label{sec:introduction}
\begin{figure*}[htbp]
\centering
\includegraphics[width=0.98\linewidth]{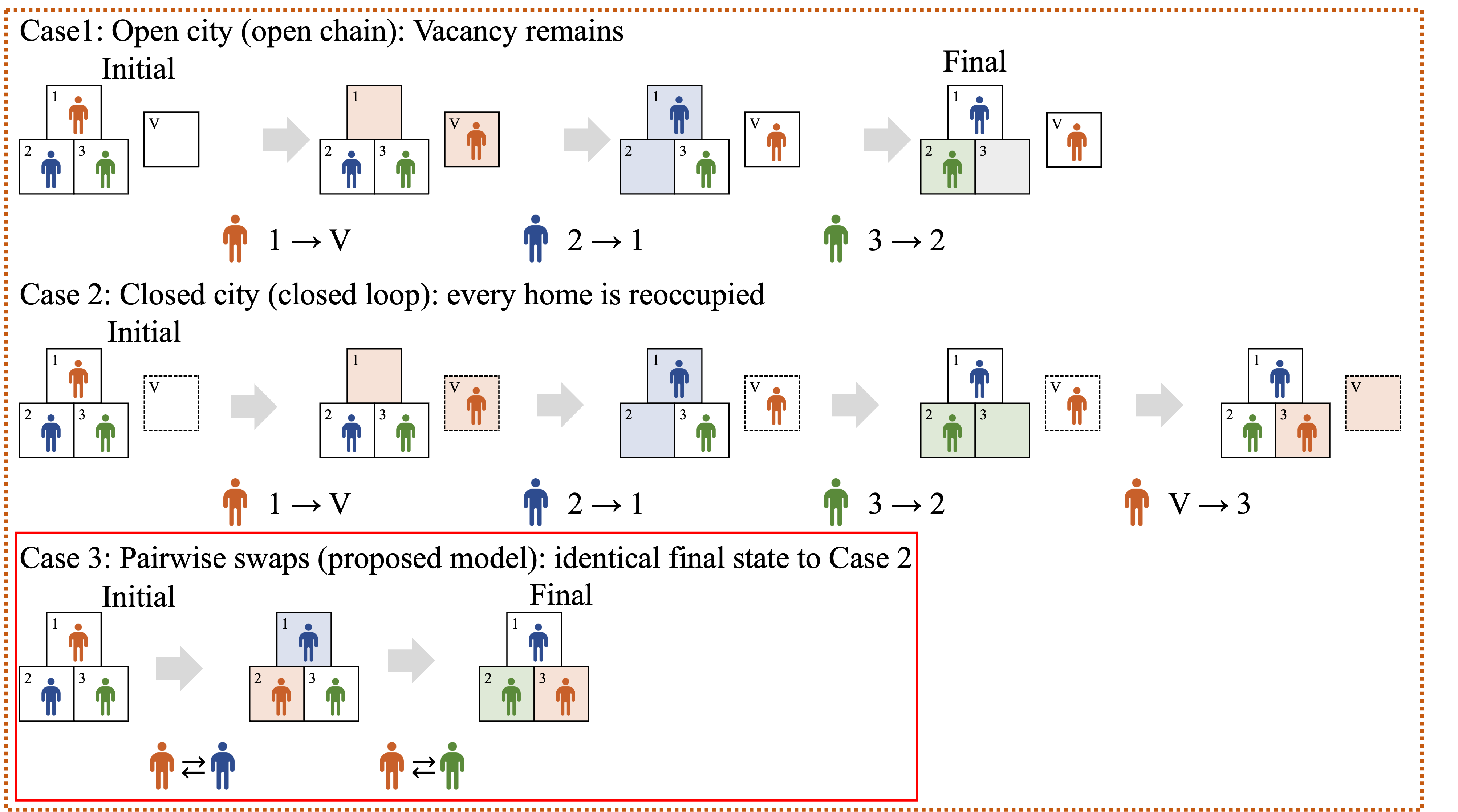}
\caption{\textbf{Pairwise origin swaps reproduce the closed-loop component of real-city relocation.}
Three residents, three homes, three workplaces (colors match residents to workplaces).
\textbf{Case~1 (Open chain):} a newly created vacancy starts a chain that ends with another home left empty; the OD row marginals change.
\textbf{Case~2 (Closed loop):} under a fixed housing stock the chain closes on itself; all row and column sums of the OD matrix are conserved.
\textbf{Case~3 (GOSM 2-swaps):} the same closed loop decomposes into two sequential pairwise swaps that yield the identical final state.
The GOSM models Case~2 via Case~3.}
\label{concept}
\end{figure*}

Human mobility and urban spatial structure are coupled through the origin--destination (OD) matrix $\mathbf{f}=\{f_{ij}\}$, which counts the commuters living in zone $i$ and working in zone $j$ \cite{wilson1967,fotheringham1989}. To describe these spatial interactions, the gravity model has long served as a standard framework across transport planning, international trade, and epidemiology \cite{wilson1967,fotheringham1989,erlander1990,tinbergen1962,anderson2003,balcan2009,lenormand2016,hilton2020,zheng2024plosone,zheng2026jsp}. In its widely used doubly constrained form \cite{wilson1967,wilson1970,erlander1990}, the gravity model predicts commuting flows under three macroscopic conditions: residential populations fixed by the housing stock, workplace populations fixed by employment capacity, and a fixed total commuting cost. Despite its success in predicting aggregate flows, this macroscopic formulation leaves the underlying individual-level mechanism unspecified.

Cities are widely regarded as systems that evolve far from equilibrium rather than settle into a static optimum \cite{allen1979,batty2008,portugali2000,bettencourt2013}. Commuting illustrates this: observed journeys are systematically longer than the cost-minimizing assignment of homes to jobs \cite{hamilton1982,horner2002}. The gravity model, however, which is one of the most widely used descriptions of commuting flows, is an equilibrium model; we call the state in which a city's commuting fits this model the gravity equilibrium. Whether residential relocation actually brings a city's commuting to the gravity equilibrium has not been tested.

The gravity equilibrium has a precise meaning. Wilson derived the doubly constrained gravity model from the maximum-entropy principle \cite{wilson1967,wilson1970,jaynes1957}, showing that it is the most probable static flow configuration compatible with fixed housing, job, and travel-cost constraints \cite{boltzmann1877}. In reality, commuting patterns evolve as households relocate to reduce their travel burden, so applying the static gravity model implicitly assumes that repeated moves have already driven the city to this equilibrium. To uncover the mechanism behind such regularities, physics-inspired frameworks have linked mobility to metabolic energy laws in transport-mode choice \cite{kolbl2003}, stable travel-time budgets \cite{schafer2000}, and network scaling \cite{barthelemy2011,barthelemy2019}. Schl\"apfer et al.\ \cite{schlapfer2021} showed that the number of visitors to a location decays as the inverse square of the product of travel distance and visiting frequency, interpreted this collapse as a conserved travel-energy budget, and reproduced it with an individual-level model of exploration and preferential return. The doubly constrained gravity model and its maximum-entropy assumption, however, remain unexplained from a dynamical, physical standpoint.

To investigate this question, we introduce a dynamical framework, the Gravity-based Origin Swapping Model (GOSM), which models residential relocation toward the gravity-model equilibrium. The model is a simplified representation of reality; the pairwise swap captures the closed-loop exchange component of real-world relocation \cite{white1970,chase1991} [Fig.~\ref{concept}]. Combined with a Metropolis rule \cite{metropolis1953,hastings1970}, swaps that reduce total commuting cost are always accepted, whereas cost-increasing swaps are accepted with a probability governed by a behavioral tolerance parameter, the temperature $T$, which in discrete choice theory corresponds to the randomness of individual decisions \cite{mcfadden1974,anas1983}. We treat the swapping process as a conserved Kawasaki dynamics \cite{kawasaki1966}: the exchange preserves every row and column sum of the OD matrix, reflecting the assumption that the housing stock and job capacities change slowly in the short run. Because these capacity-preserving swaps satisfy detailed balance and connect all admissible OD matrices \cite{ryser1957,diaconis1998,kleitman1973}, the GOSM is an equilibrium-state generator, whose unique stationary state is the doubly constrained gravity model \cite{newman1999}.

Our main finding is that a city whose houses are all occupied, as in the GOSM, shows non-equilibrium properties. We calibrate the model to 30 million GPS commuting trajectories from six Japanese metropolitan areas \cite{zheng2024plosone,zheng2026jsp}, bringing the simulated cities as close to real ones as the data allow, and map a phase diagram that gives the conditions under which equilibrium is and is not reached; the parameters of the six cities fall in the non-equilibrium regime. There, instead of converging to the lowest free-energy state, the relaxation exhibits three signatures characteristic of glassy dynamics: broken ergodicity across different initial residential assignments \cite{palmer1982,BerthierBiroli2011}, aging \cite{bouchaud1997,abella2022}, and hysteresis \cite{BerthierBiroli2011,Angell1995}. Because the free-energy landscape is a smooth single well, this freezing is dynamical rather than thermodynamic in origin, placing the GOSM among the kinetically constrained models (KCMs) of glasses \cite{RitortSollich2003,FredricksonAndersen1984,jackle1991,ChandlerGarrahan2010}.

The remainder of this paper is organized as follows. Section~\ref{metho:sec_methods} describes the empirical commuting dataset, formulates the GOSM relocation dynamics alongside the classical gravity model, and calibrates the behavioral parameters to the Tokyo metropolitan area. Section~\ref{result:sec_results} first proves the strict convexity of the urban free-energy landscape and the equivalence between the stationary state of the GOSM and the doubly constrained gravity model (Sec.~\ref{result:subsec_equivalence}), and then examines the nonequilibrium glassy arrest of the relocation dynamics through ergodicity breaking, aging, hysteresis, and free-energy trajectories (Secs.~\ref{result:subsec_failure}--\ref{result:subsec_landscape}). Finally, Sec.~\ref{sec:discussion} discusses the connection between the GOSM and real-world relocation processes, and the implications for urban transport policy more broadly.


\section{Model and methods}
\label{metho:sec_methods}

This section establishes the empirical data, the theoretical models, and the simulation protocols that underpin the results of Section~\ref{result:sec_results}. The mobility dataset (Sec.~\ref{metho:subsec_data}) provides the empirical origin--destination (OD) matrices and travel times of six major Japanese metropolitan areas. We then formulate two complementary routes to the theoretical equilibrium commuting pattern (Sec.~\ref{metho:subsec_routes}): Wilson's static fixed-cost gravity model (Sec.~\ref{metho:subsec_gravity}) and our dynamical Gravity-based Origin Swapping Model (GOSM, Sec.~\ref{metho:subsec_gosm}), which together define the urban free energy analyzed in Sec.~\ref{result:subsec_equivalence}. Next, we define the urban reference states and the job--housing proximity order parameter $\phi_{\mathrm{norm}}$ (Sec.~\ref{metho:subsec_reference}) used to track relocation trajectories in Secs.~\ref{result:subsec_failure}--\ref{result:subsec_landscape}. Finally, we calibrate the behavioral parameters $(\gamma^{*}, T^{*})$ of each metropolitan area (Sec.~\ref{metho:subsec_calibration}) and specify the Monte Carlo simulation protocols (Sec.~\ref{metho:subsec_protocols}); further details of the data processing are given in Appendix~\ref{app:data}, and the derivations omitted from the main text in Appendix~\ref{app:derivations}.

\subsection{Data}
\label{metho:subsec_data}
To ground the framework empirically, we analyze anonymized smartphone GPS records from six major Japanese metropolitan areas, Tokyo, Osaka, Nagoya, Fukuoka, Sapporo, and Sendai, each with a population exceeding 2 million \cite{estat2024}. Each metropolitan area is partitioned into 1~km $\times$ 1~km grid cells, yielding approximately 10{,}000 populated spatial units for Tokyo (Table~\ref{app_tab:data}). The raw mobility dataset, provided by Agoop Corp., covers roughly one million users nationwide per day with a spatial accuracy of about 10~m, and each record contains a daily-randomized user ID, a timestamp, coordinates, and municipality-level home and work codes. To focus on regular weekday commuting, we restrict the sample to working days in 2023 and retain active users with more than 100 GPS records per day whose home and workplace municipalities both lie within the target metropolitan area.

From these trajectories, we construct the empirical commuting OD matrix $\mathbf{f} = \{f_{ij}\}$ by identifying each user's residential and workplace cells. We first detect home and workplace locations at a 100~m resolution, defining the home as the cell occupied at 5~a.m.\ with an accumulated stay exceeding 4 hours, and the workplace as a distinct cell with at least 5 hours of daytime stay. For the $\sim$3\% of users with multiple candidate workplaces, we select the first workplace visited after leaving home, which ensures a strict one-to-one home--work assignment per commuter. Under these criteria, approximately 60\% of users are classified as commuters, yielding $M = \sum_{ij} f_{ij} \approx 30$ million daily commuting trajectories over the year (14.9 million in Tokyo alone; Table~\ref{app_tab:data}), broadly consistent with official statistics \cite{statbureau_commuter} given the known demographic bias of smartphone data \cite{zheng2024plosone,zheng2026jsp}. Because user IDs are randomized nightly for privacy, trajectories cannot be linked across days, and each daily home--work pair is treated as an independent observation in the aggregate OD matrix.

The inter-zone travel time $t_{ij}$ is estimated from the observed duration of morning commutes between residential zone $i$ and workplace zone $j$. For each trajectory, we measure the shortest elapsed time between the last GPS record at home and the first record at work. The resulting travel-time distribution decays approximately exponentially (SM Fig.~S1), and we cap commuting times at 120 minutes, which retains over 99.5\% of all recorded trips. In the Tokyo metropolitan area, this procedure yields a mean commuting time of 39.9 minutes, slightly below the 45.9 minutes reported for household heads in census surveys \cite{statbureau_commuter}, because our sample also captures younger and part-time workers who typically reside closer to their jobs.

\subsection{Two routes to the equilibrium OD matrix}
\label{metho:subsec_routes}
To test whether residential mobility can reach the gravity-model prediction, we formulate two complementary statistical routes that share the same theoretical equilibrium OD matrix $\mathbf{f}^{*}$ (Fig.~\ref{fig1}). The first route is static: Wilson's doubly constrained gravity model \cite{wilson1967,wilson1970}, which treats the city as having a strictly fixed total commuting-time budget (the energy $E$ in this study, Eq.~\eqref{eq:energy}; analogous to an isolated, microcanonical system in physics) and identifies $\mathbf{f}^{*}$ by maximizing the configurational entropy. The second route is dynamical: the GOSM, a stylized Metropolis relocation process \cite{metropolis1953,hastings1970} in which pairs of commuters swap residences under a temperature $T$ (analogous to a canonical system regulated by a heat bath). Because the maximum-entropy macrostate of Wilson's static model coincides with the most probable stationary state of the GOSM, the two routes are equivalent in large cities. This equivalence allows us to interpret the classical gravity model as the stationary equilibrium of an explicit residential relocation process, and to diagnose departures from it as dynamical arrest.

To keep the text accessible to readers in both urban studies and statistical physics, Table~\ref{tab:notation} summarizes all mathematical symbols together with their physical and urban interpretations.

\begin{table}[t]
  \caption{\label{tab:notation}Notation and its urban meaning, grouped by model component.}
  \small
  \begin{ruledtabular}
  \begin{tabular}{@{}l@{\hspace{6pt}}p{0.62\columnwidth}@{}}
    \multicolumn{2}{@{}l}{\emph{Shared variables}} \\
    $\mathbf{f} = \{f_{ij}\}$ & Origin--destination (OD) matrix \\
    $f_{ij}$ & Commuters living in origin $i$, working in destination $j$ \\
    $M = \sum_{ij} f_{ij}$ & Total number of commuters \\
    $h_i = \sum_j f_{ij}$ & Residents in origin zone $i$ (fixed) \\
    $w_j = \sum_i f_{ij}$ & Jobs in destination zone $j$ (fixed) \\
    $t_{ij}$ & Travel time from $i$ to $j$ (fixed) \\
    $\gamma$ & Exponent of perceived travel time, Eq.~\eqref{eq:energy} \\
    $E[\mathbf{f}]$ & Total perceived commuting cost (energy), Eq.~\eqref{eq:energy} \\
    $T$, $\beta = 1/T$ & Temperature (tolerance for a longer commute) \\
    $\mathbf{f}^{*}$ & Gravity equilibrium OD matrix, Eq.~\eqref{eq:gravity} \\
    $\Gamma$ & Set of all OD matrices with fixed $\{h_i\}$ and $\{w_j\}$ \\[3pt]
    \multicolumn{2}{@{}l}{\emph{Static gravity model (Sec.~\ref{metho:subsec_gravity})}} \\
    $\Omega(\mathbf{f})$ & Number of individual assignments giving $\mathbf{f}$ \\
    $S(\mathbf{f}) = \ln \Omega(\mathbf{f})$ & Entropy: diversity of origin--destination arrangements \\
    $\mathcal{F}(\mathbf{f}) = E - TS$ & Free energy: cost--diversity trade-off \\
    $A_i, B_j$ & Balancing factors enforcing $h_i$ and $w_j$ \\
    $F(t)$ & Deterrence function, Eq.~\eqref{eq:deterrence} \\[3pt]
    \multicolumn{2}{@{}l}{\emph{GOSM relocation dynamics (Sec.~\ref{metho:subsec_gosm})}} \\
    $\pi(\mathbf{f})$ & Stationary probability of observing OD matrix $\mathbf{f}$ \\
    $\phi_{\mathrm{norm}}$ & Job--housing proximity index (order parameter), Eq.~\eqref{eq:order-parameter} \\
    $C(\tau, t_w)$ & Fraction of commuters not relocated between $t_w$ and $t_w+\tau$, Eq.~\eqref{eq:two-time} \\
  \end{tabular}
  \end{ruledtabular}
\end{table}

\begin{figure}[t]
\centering
\begin{tikzpicture}[
  font=\footnotesize,
  boxstep/.style={draw=black!50, rounded corners=1pt, align=center,
                inner sep=3pt, text width=3.55cm, minimum height=8mm},
  boxfinal/.style={draw=black!80, semithick, rounded corners=1pt,
                 align=center, inner sep=4pt, text width=7.0cm},
  arrflow/.style={-{Stealth[length=1.6mm]}, black!60}
]
\node[align=center, font=\footnotesize\bfseries] (ha) at (0,0)
  {Gravity equilibrium\\ (Wilson)};
\node[align=center, font=\footnotesize\bfseries] (hb) at (4.1,0)
  {GOSM\\ (this study)};
\node[boxstep] (a1) at (0,-1.0)
  {maximize entropy subject to:};
\node[boxstep] (a2) at (0,-2.15)
  {fixed homes and jobs $h_i, w_j$;\\ fixed total commuting time cost};
\node[boxstep] (a3) at (0,-3.3)
  {parameter calibration)};
\node[boxstep] (b1) at (4.1,-1.0)
  {origin-swap move Eq.~\eqref{eq:swap}\\ (marginals preserved)};
\node[boxstep] (b2) at (4.1,-2.15)
  {$P_{\mathrm{acc}}=\min(1,e^{-\beta\Delta E})$;\\
   $\beta$: inverse temperature (an \emph{input})};
\node[boxstep] (b3) at (4.1,-3.3)
  {detailed balance $+$ coarse-\\graining $\Rightarrow$
   $\pi(\mathbf f)\propto e^{-\beta\mathcal F(\mathbf f)}$};
\draw[arrflow] (a1) -- (a2);
\draw[arrflow] (a2) -- (a3);
\draw[arrflow] (b1) -- (b2);
\draw[arrflow] (b2) -- (b3);
\node[boxfinal] (same) at (2.05,-5.05)
  {same equilibrium OD matrix:\quad
   $f^{*}_{ij}=A_i h_i\, B_j w_j\, e^{-\beta t_{ij}^{\gamma}}$\\[2pt]};
\draw[arrflow] (a3.south) -- ($(same.north)+(-1.6,0)$);
\draw[arrflow] (b3.south) -- ($(same.north)+(+1.6,0)$);
\end{tikzpicture}
\caption{\textbf{Two independent routes converge on the same equilibrium OD matrix.} Wilson's static variational formulation fixes the city's total commuting-time budget (the energy $E$ in this study) exactly (a microcanonical description), whereas the dynamical GOSM regulates commuting costs through a temperature $T = \beta^{-1}$ (a macrocanonical description). Left: maximizing the entropy under fixed homes, jobs, and total cost yields the doubly constrained gravity model [Eq.~\eqref{eq:gravity}], with $\beta$ emerging as a Lagrange multiplier. Right: pairwise origin swaps preserve the zone capacities exactly, and detailed balance yields the stationary distribution [Eq.~\eqref{eq:canonical}], whose mode is the expected to be the same equilibrium matrix $\mathbf{f}^{*}$. Strict concavity of the entropy density guarantees ensemble equivalence in large systems \cite{touchette2015}.}
\label{fig1}
\end{figure}

\subsection{Doubly constrained gravity model}
\label{metho:subsec_gravity}

\subsubsection{Fixed homes and jobs}
\label{metho:subsubsec_constraints}
The macroscopic state of the urban commuting system is represented by the non-negative integer OD matrix $\mathbf{f} = \{f_{ij}\}$, which counts the commuters residing in zone $i$ and working in zone $j$. Every admissible commuting configuration satisfies the hard marginal constraints
\begin{equation}
  \sum_j f_{ij} = h_i \quad \text{(homes)}, \qquad \sum_i f_{ij} = w_j \quad \text{(jobs)},
  \label{eq:marginals}
\end{equation}
which fix the row and column sums of $\mathbf{f}$ to the observed residential population $h_i$ and workplace employment $w_j$. Mathematically, Eq.~\eqref{eq:marginals} restricts the feasible state space $\Gamma$ to matrices with prescribed marginals. In urban terms, the condition reflects the short-run rigidity of the built environment: housing stock and job locations are fixed, and only the matching between homes and workplaces can be rearranged.

\subsubsection{Commuting cost and deterrence function}
\label{metho:subsubsec_deterrence}

\begin{figure*}[htbp]
\centering
\includegraphics[width=0.95\linewidth]{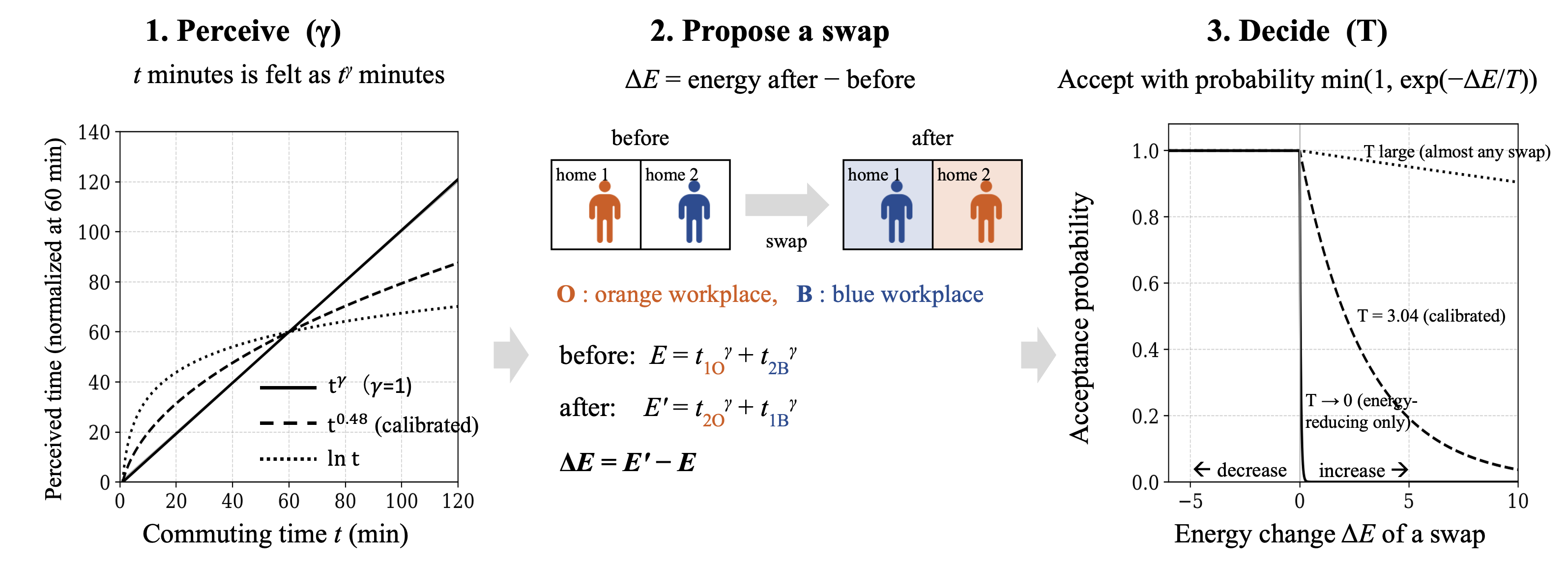}
\caption{\textbf{Commuters perceive travel time nonlinearly and barely accept small commuting time increases.}
(1)~Perceived versus actual commuting time: the calibrated power-law model ($t^{0.48}$) and the logarithmic model both show that an extra minute feels shorter on an already long commute than on a short one.
(2)~A proposed swap exchanges the homes of two commuters, changing the total energy by $\Delta E$.
(3)~The swap is accepted with probability $\min(1, e^{-\Delta E/T})$: at $T \to 0$ only cost-reducing swaps proceed; at large $T$ almost any swap is accepted; $T = 3.04$ is the value calibrated for Tokyo.}
\label{fig:gamma}
\end{figure*}

To quantify the aggregate travel burden of the city, each OD matrix $\mathbf{f}$ is assigned an energy-like total commuting cost,
\begin{equation}
  E[\mathbf{f}] \;=\; \sum_{ij} f_{ij}\, t_{ij}^{\gamma}, \qquad \gamma > 0,
  \label{eq:energy}
\end{equation}
which sums the perceived travel cost $t_{ij}^{\gamma}$ over all commuters. Although $E[\mathbf{f}]$ plays the formal role of the system energy in our statistical-mechanical framework, it represents the aggregate travel disutility that residents collectively seek to reduce, not a physical energy. The exponent $\gamma$ addresses a long-standing debate in the spatial-interaction literature over how commuting flows decay with travel time through a deterrence function $F(t)$, with $f_{ij} \propto F(t_{ij})$ \cite{fotheringham1989,barthelemy2011,lenormand2016}. Specifically, the cost specification \eqref{eq:energy} generates the stretched-exponential deterrence kernel
\begin{equation}
  F(t) \;=\; e^{-\beta t^{\gamma}},
  \label{eq:deterrence}
\end{equation}
which nests both classical deterrence forms as special cases. As shown empirically in Sec.~\ref{metho:subsubsec_poisson}, this unified kernel captures the observed decay of commuter flows with travel time better than either classical alternative [Fig.~\ref{method:fig3}(a)].


Mathematically, the stretched exponential bridges the exponential and power-law regimes as $\gamma$ varies between $0$ and $1$. At $\gamma = 1$, Eq.~\eqref{eq:deterrence} reduces to the simple exponential deterrence $F(t) = e^{-\beta t}$; in the opposite limit $\gamma \to 0$ (with $\beta\gamma$ held finite), it reduces to a power law (Appendix~\ref{app:limits}). For intermediate values $0 < \gamma < 1$, the deterrence behaves like a power law at moderate travel times while remaining exponentially bounded for extreme commutes. Behaviorally [Fig.~\ref{fig:gamma}], $t^{\gamma}$ represents the perceived cost of a trip of duration $t$, and $\gamma < 1$ indicates diminishing marginal sensitivity to extra minutes on an already long commute.

\subsubsection{Equilibrium OD matrix}
\label{metho:subsubsec_entropy}
Substituting the perceived cost $t^{\gamma}$ into Wilson's maximum-entropy derivation \cite{wilson1967,wilson1970} gives the equilibrium solution of the doubly constrained gravity model,
\begin{equation}
  f^{*}_{ij} \;=\; A_i h_i \, B_j w_j \, e^{-\beta t_{ij}^{\gamma}},
  \label{eq:gravity}
\end{equation}
where the balancing factors $A_i$ and $B_j$ enforce the marginal constraints \eqref{eq:marginals} and are computed by iterative proportional fitting (IPF) \cite{furness1965,deming1940,sinkhorn1967}. In this static construction, $\beta$ emerges as the Lagrange multiplier of the cost constraint and is tuned until the aggregate budget $E[\mathbf{f}^{*}]$ is met. The derivation and the balancing equations are given in Appendix~\ref{app:maxent}.

\subsubsection{Free-energy formulation}
\label{metho:subsubsec_landau}
In statistical-mechanical terms, Wilson's fixed-cost formulation corresponds to a microcanonical ensemble \cite{sagarra2013,erlander1990}: the urban system is treated as isolated, with its aggregate commuting cost fixed at $E[\mathbf{f}] = E_{\mathrm{obs}}$ alongside the hard housing and job constraints \eqref{eq:marginals}. The quantity maximized is the configurational entropy, which counts the number of ways the $M$ individual commuters can be assigned to the cell counts $\{f_{ij}\}$,
\begin{equation}
  S(\mathbf{f}) \;=\; \ln \Omega(\mathbf{f}), \qquad
  \Omega(\mathbf{f}) \;=\; \frac{M!}{\prod_{ij} f_{ij}!},
  \label{eq:multiplicity}
\end{equation}
which by Stirling's approximation reads $S(\mathbf{f}) \simeq -\sum_{ij} f_{ij}\ln f_{ij} + \mathrm{const}$ \cite{boltzmann1877,jaynes1957}. In urban terms, $S(\mathbf{f})$ measures the diversity of residential arrangements compatible with a given flow pattern.

This constrained entropy maximization can be equivalently reformulated as the minimization of an urban free energy \cite{callen1985,erlander1990}. For any admissible OD matrix $\mathbf{f}$, we define the Landau free energy
\begin{equation}
  \mathcal{F}(\mathbf{f}) \;=\; E[\mathbf{f}] - T\, S(\mathbf{f}),
  \label{eq:free-energy}
\end{equation}
where $T = 1/\beta$ is the temperature and $S(\mathbf{f})$ is the configurational entropy of Eq.~\eqref{eq:multiplicity}. In an urban context, $\mathcal{F}(\mathbf{f})$ quantifies the macroscopic trade-off between economizing aggregate commuting time, $E[\mathbf{f}]$, and accommodating diverse residential locations, $T S(\mathbf{f})$. Maximizing $S$ at a fixed commuting-time budget $E_{\mathrm{obs}}$ is mathematically identical to minimizing $\mathcal{F}$ at a fixed temperature $T$, which places the gravity-model state $\mathbf{f}^{*}$ at the global minimum of the free-energy landscape (Proposition~\ref{prop:convex}).

\subsection{The Gravity-based Origin Swapping Model}
\label{metho:subsec_gosm}
To model how individual residential relocations drive a city toward this equilibrium, we introduce a stylized dynamical framework, the Gravity-based Origin Swapping Model (GOSM). Because the housing stock $h_i$ and the job capacities $w_j$ are fixed in the short run, the minimal elementary move that preserves all zone capacities is a pairwise origin swap (Fig.~\ref{concept}, Case~3): two randomly chosen commuters $a$ and $b$ exchange their residences while keeping their respective workplaces,
\begin{equation}
  (i_a, j_a),\, (i_b, j_b) \;\longrightarrow\; (i_b, j_a),\, (i_a, j_b).
  \label{eq:swap}
\end{equation}
Although simplified relative to full housing-market transactions, this Kawasaki-type conserved exchange captures the closed-loop component of real-world residential vacancy chains \cite{white1970,chase1991} while satisfying the marginal constraints \eqref{eq:marginals} exactly at every step. The proposed swap changes the combined perceived commuting cost of the two households by
\begin{equation}
  \Delta E \;=\; t_{i_b j_a}^{\gamma} + t_{i_a j_b}^{\gamma} - t_{i_a j_a}^{\gamma} - t_{i_b j_b}^{\gamma},
  \label{eq:deltaE}
\end{equation}
and the relocation is accepted with the standard Metropolis probability \cite{metropolis1953},
\begin{equation}
  P_{\mathrm{acc}} \;=\; \min\bigl(1,\, e^{-\beta \Delta E}\bigr).
  \label{eq:metropolis}
\end{equation}

Swaps that reduce the total commuting cost ($\Delta E \le 0$) are always accepted, whereas cost-increasing swaps ($\Delta E > 0$) are accepted with a probability controlled by the temperature $T = 1/\beta$, an input parameter of the model whose role is illustrated in Fig.~\ref{fig:gamma}.

This pairwise origin-swapping rule has a well-defined stationary distribution, as follows from the chain of reasoning below.

\smallskip\noindent
(i)~\emph{Symmetric proposal $\Rightarrow$ detailed balance.}
Drawing a pair $(a,b)$ and proposing an origin swap is symmetric with proposing the reverse, so the Metropolis rule \eqref{eq:metropolis} satisfies detailed balance between any two individual commuter assignments.

\smallskip\noindent
(ii)~\emph{Boltzmann weight over assignments.}
Consequently, every individual assignment realizing a given OD matrix $\mathbf{f}$ is sampled with a stationary weight proportional to $e^{-\beta E[\mathbf{f}]}$.

\smallskip\noindent
(iii)~\emph{Connectivity.}
A finite sequence of pairwise swaps connects all assignments compatible with the fixed origin and destination capacities \cite{ryser1957,kleitman1973,newman1999}, so the chain is ergodic.

\smallskip\noindent
(iv)~\emph{Coarse-graining $\Rightarrow$ canonical distribution.}
Exactly $\Omega(\mathbf{f})$ distinct assignments produce the same flow matrix $\mathbf{f}$ and share the same cost $E[\mathbf{f}]$, so the stationary probability of observing $\mathbf{f}$ is
\begin{equation}
  \pi(\mathbf{f}) \;\propto\; \Omega(\mathbf{f})\, e^{-\beta E[\mathbf{f}]} \;=\; e^{-\beta \mathcal{F}(\mathbf{f})},
  \label{eq:canonical}
\end{equation}
where $\mathcal{F}(\mathbf{f}) = E[\mathbf{f}] - T S(\mathbf{f})$ is the urban free energy of Eq.~\eqref{eq:free-energy}. The local acceptance rule \eqref{eq:metropolis} contains no explicit entropy term; the multiplicity $\Omega(\mathbf{f})$ emerges from the coarse-graining.

\smallskip\noindent
(v)~\emph{Mode $=$ gravity model.}
Maximizing $\pi(\mathbf{f})$ is equivalent to minimizing $\mathcal{F}(\mathbf{f})$ and yields the same stationarity condition as Wilson's derivation (Appendix~\ref{app:derivations}), so the GOSM and the static gravity model single out the same equilibrium matrix $\mathbf{f}^{*}$.

\subsection{Reference states and order parameter}
\label{metho:subsec_reference}
To evaluate where an observed city sits and how the relocation dynamics evolves, we compare four benchmark OD matrices that share exactly the same housing and job marginals \eqref{eq:marginals}: the observed city $\mathbf{f}^{\mathrm{obs}}$, the theoretical gravity equilibrium $\mathbf{f}^{*}$, the minimum-cost ground state $\mathbf{f}^{\mathrm{grnd}}$, and the spatially indifferent random state $\mathbf{f}^{\mathrm{rand}}$. In our framework, the ground state represents the zero-temperature limit ($T \to 0$), in which commuters strictly minimize travel time, whereas the random state represents the infinite-temperature limit ($T \to \infty$), in which residential choices are insensitive to commute duration. As the temperature $T$ varies, the theoretical equilibrium $\mathbf{f}^{*}$ interpolates between these two extremes. Tracking how the GOSM dynamics moves the system among these four states forms the core of the analysis in Section~\ref{result:sec_results}.

\subsubsection{Random state}
\label{metho:subsubsec_random}
The random state $\mathbf{f}^{\mathrm{rand}}$ represents a baseline city in which commuters are matched to residences and workplaces without regard to travel time ($T \to \infty$). Mathematically, it is the maximum-entropy OD matrix subject only to the housing and job capacities, with flows proportional to the product of the marginals, $f_{ij} \propto h_i w_j$. We construct it by applying IPF to the seed $h_i w_j$ and then rounding the continuous entries to an integer matrix that preserves every row and column sum exactly.

\subsubsection{Ground state}
\label{metho:subsubsec_ground}
At the opposite extreme, the ground state $\mathbf{f}^{\mathrm{grnd}}$ represents the theoretical minimum-commute city ($T \to 0$), in which residences and workplaces are matched so as to minimize the total travel burden. Formally, it is the global cost minimizer
\begin{equation}
  \mathbf{f}^{\mathrm{grnd}} \;=\; \arg\min_{\mathbf{f}} E[\mathbf{f}]
  \label{eq:ground-def}
\end{equation}
subject to the fixed housing and job capacities \eqref{eq:marginals}. We solve it exactly as a minimum-cost network flow \cite{ahuja1993}, in which residential zones supply $h_i$ commuters, workplace zones demand $w_j$, and each commuter routed from $i$ to $j$ costs $t_{ij}^{\gamma}$.

\subsubsection{Order parameter}
\label{metho:subsubsec_order}
To summarize the spatial sorting of commuters along a single scalar dimension, we define an urban order parameter $\phi_{\mathrm{norm}}$ that measures the normalized share of commuters living in the immediate vicinity of their workplace. We first define the set of local, near-zero-cost commuting pairs as those within the same grid cell or with a travel time below $t_c = 5$~minutes,
\begin{equation}
  \mathcal{O} \;=\; \{(i,j) : i = j \;\text{ or }\; t_{ij} \le t_c\}, \qquad t_c = 5~\mathrm{min},
  \label{eq:ordered-cell}
\end{equation}
and normalize the raw fraction of such short-distance commuters between the random-state and ground-state benchmarks:
\begin{align}
  \phi_{\mathrm{raw}}(\mathbf{f}) &\;=\; \frac{1}{M} \sum_{(i,j) \in \mathcal{O}} f_{ij}, \notag \\
  \phi_{\mathrm{norm}} &\;=\; \frac{\phi_{\mathrm{raw}}(\mathbf{f}) - \phi_{\mathrm{raw}}(\mathbf{f}^{\mathrm{rand}})}{\phi_{\mathrm{raw}}(\mathbf{f}^{\mathrm{grnd}}) - \phi_{\mathrm{raw}}(\mathbf{f}^{\mathrm{rand}})}.
  \label{eq:order-parameter}
\end{align}
By construction, $\phi_{\mathrm{norm}} = 0$ corresponds to the spatially indifferent random assignment and $\phi_{\mathrm{norm}} = 1$ to the maximally compact ground state. As commuters become more sensitive to travel time (lower $T$), residential sorting drives flows toward the short-trip cells $\mathcal{O}$, so $\phi_{\mathrm{norm}}$ is a direct measure of job--housing co-location.

The 5-minute threshold $t_c$ eliminates discretization artifacts arising from arbitrary grid boundaries. Relying on the matrix diagonal alone ($i = j$) would make the measured share of local commuters artificially sensitive to zone size and grid placement, the well-known modifiable areal unit problem (MAUP) of spatial analysis \cite{openshaw1984,viegas2009}. Including all pairs with $t_{ij} \le t_c$ ensures that commuters who cross a grid boundary on a very short trip are consistently counted as living near their workplace, following standard practice in the excess-commuting literature \cite{hu2015}. Because both benchmark values $\phi_{\mathrm{raw}}(\mathbf{f}^{\mathrm{rand}})$ and $\phi_{\mathrm{raw}}(\mathbf{f}^{\mathrm{grnd}})$ in Eq.~\eqref{eq:order-parameter} are evaluated with the same set $\mathcal{O}$, the normalized index $\phi_{\mathrm{norm}}$ is a scale-consistent observable across all six metropolitan areas.

\subsection{Parameter calibration}
\label{metho:subsec_calibration}

\begin{figure*}[htpb]
\centering
\includegraphics[width=0.98\linewidth]{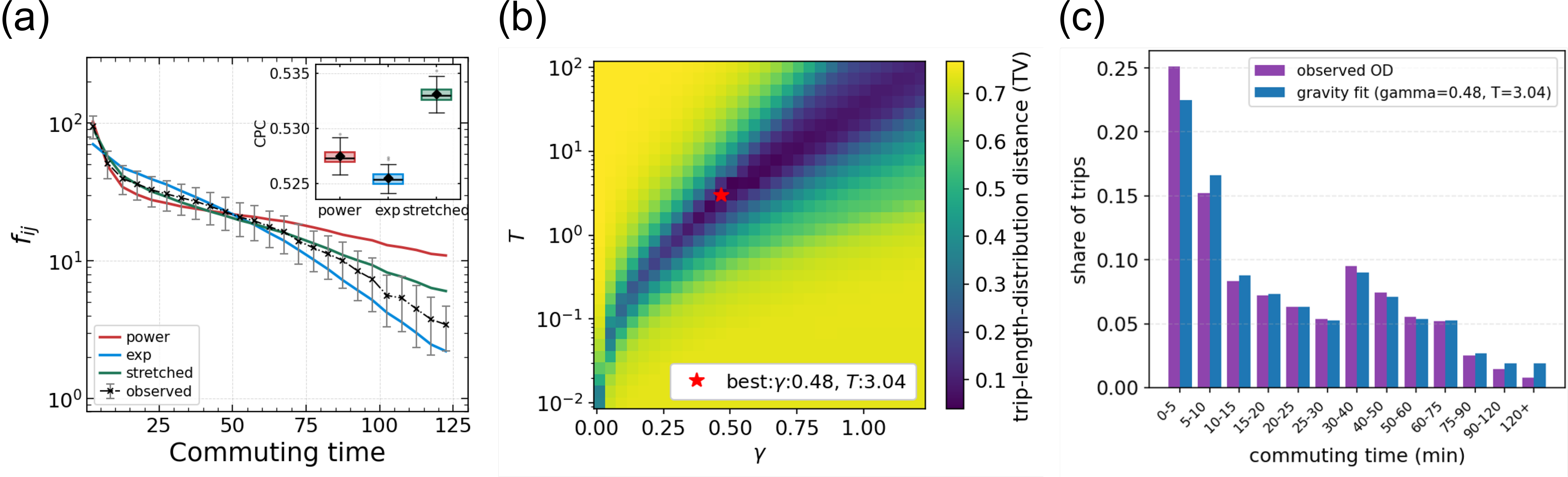}
\caption{\textbf{Empirical validation and parameter calibration of the commuting deterrence function for the Tokyo metropolitan area.} (a)~Binned commuter flow $f_{ij}$ versus travel time $t_{ij}$ (5-minute bins; markers denote medians and error bars span interquartile ranges), compared with the doubly constrained gravity model~\eqref{eq:gravity} fitted by Poisson maximum likelihood under three deterrence forms. The stretched exponential tracks the empirical decay across all durations, whereas the power law overestimates and the simple exponential underestimates long-distance commutes beyond 75 minutes. Inset: common part of commuters (CPC) over $10^{2}$ Poisson resamples of the observed OD matrix (diamonds mark full-sample values; all six cities are reported in Appendix~\ref{app:sixcities}). (b)~Total-variation distance $D_{\mathrm{tld}}$ between the observed and model trip-length distributions over the $(\gamma, T)$ plane; the global minimum selects $\gamma^{*} = 0.48$ and $T^{*} = 3.04$ (red star). The flat plateau at large $\gamma$ and small $T$ is the zero-temperature ($T \to 0$) limit, where the equilibrium collapses onto the minimum-cost ground state. (c)~Observed trip-length distribution and the calibrated equilibrium prediction at $(\gamma^{*}, T^{*})$.}
\label{method:fig3}
\end{figure*}

Applying the framework to real metropolitan areas requires determining both the functional form of the travel deterrence kernel and the city-specific behavioral parameters $(\gamma, T)$. In Sec.~\ref{metho:subsubsec_poisson}, we first validate the stretched-exponential cost form $t^{\gamma}$ against the classical exponential and power-law alternatives by fitting directly to the observed cell-by-cell OD flows. In Sec.~\ref{metho:subsubsec_heatbath}, we then calibrate the parameter pair $(\gamma^{*}, T^{*})$ by identifying the theoretical equilibrium whose trip-length distribution best matches each observed city. Although the two procedures optimize different statistical objectives on different observables, they yield mutually consistent parameters across all six metropolitan areas.

\subsubsection{Validation of the deterrence function}
\label{metho:subsubsec_poisson}
The cost specification $E = \sum_{ij} f_{ij} t_{ij}^{\gamma}$ is first tested against the data. Because the energy \eqref{eq:energy} and the deterrence function \eqref{eq:deterrence} are both determined by the same perceived-time setting $t^{\gamma}$, we test the specification by asking which deterrence function best explains the empirical OD matrices. Specifically, we test whether the stretched-exponential perceived cost $g(t) = t^{\gamma}$ ($0 < \gamma < 1$) outperforms the linear cost $g(t) = t$ (exponential deterrence) and the logarithmic cost $g(t) = \ln t$ (power-law deterrence). We model the observed cell counts by Poisson regression within the doubly constrained framework and maximize the profile likelihood over $(\gamma, \beta)$. Model performance is evaluated by the Akaike information criterion (AIC) together with the common part of commuters (CPC) \cite{lenormand2016},
\begin{equation}
  \mathrm{CPC} \;=\; \frac{\sum_{ij} \min\bigl(f^{\mathrm{obs}}_{ij},\, f^{\mathrm{model}}_{ij}\bigr)}{M},
  \label{eq:cpc}
\end{equation}
which measures the fraction of the $M$ commuting trajectories correctly reproduced by the model (from $0$ for disjoint matrices to $1$ for a perfect match). As shown for Tokyo in Fig.~\ref{method:fig3}(a), the stretched-exponential kernel follows the empirical decay of commuting flows over the full 120-minute range, whereas the power law overpredicts and the simple exponential underpredicts long-distance trips. Across all six metropolitan areas, the stretched exponential achieves the highest CPC and the lowest AIC, with fitted exponents in the range $\gamma \in [0.40, 0.50]$ (Appendix~\ref{app:sixcities}, Fig.~\ref{app_fig_poisson} and Tables~\ref{app_tab:pois_fit}--\ref{app_tab:model_comparison}), confirming the cost specification \eqref{eq:energy}.

\subsubsection{Calibration of $(\gamma, T)$ to the observed city}
\label{metho:subsubsec_heatbath}
With the stretched-exponential cost form established, the parameter pair $(\gamma, T)$ is calibrated by identifying the theoretical gravity equilibrium $\mathbf{f}^{\mathrm{grav}}(\gamma, T)$ that corresponds most closely to each observed city. For each candidate pair on a two-dimensional grid (linear in $\gamma$, logarithmic in $T$), the deterrence kernel $e^{-t_{ij}^{\gamma}/T}$ is fixed and the IPF balancing equations [Eq.~\eqref{eq:balancing}] are iterated with the empirical marginals $h_i, w_j$. Substituting the result into Eq.~\eqref{eq:gravity} yields the unique equilibrium OD matrix $\mathbf{f}^{\mathrm{grav}}(\gamma, T)$ that the GOSM dynamics would reach if it equilibrated at temperature $T$ and perception exponent $\gamma$. The city parameters $(\gamma^{*}, T^{*})$ are then selected to minimize the discrepancy between this theoretical equilibrium and the empirical pattern $\mathbf{f}^{\mathrm{obs}}$.

Closeness between the observed and equilibrium matrices is measured through the commuting trip-length distribution (TLD), a standard macroscopic calibration target for spatial-interaction models \cite{hyman1969}. Travel times are grouped into 5-minute bins, and the total-variation distance $D_{\mathrm{tld}}(\gamma, T)$ between the empirical and theoretical bin shares is minimized. The explicit objective is given in Appendix~\ref{app:tld}, and a robustness check against a cell-level matrix distance is reported in SM Sec.~S3.

Figure~\ref{method:fig3}(b) displays the resulting calibration landscape $D_{\mathrm{tld}}(\gamma, T)$ for the Tokyo metropolitan area. Because $\gamma$ and $T$ enter the deterrence kernel only through the ratio $t^{\gamma}/T$, the low-error region forms a narrow diagonal valley in which a stronger nonlinearity $\gamma$ is compensated by a higher temperature $T$. Within this valley, the global minimum uniquely identifies Tokyo's calibrated parameters as $\gamma^{*} = 0.48$ and $T^{*} = 3.04$, while the flat plateau at large $\gamma$ and small $T$ reflects the zero-temperature ($T \to 0$) ground-state limit. At $(\gamma^{*}, T^{*})$ the theoretical equilibrium reproduces the observed commuting-time distribution closely across all duration bins [Fig.~\ref{method:fig3}(c)] and agrees with the independent Poisson likelihood estimates ($\gamma = 0.40$, $T = 2.96$; Table~\ref{app_tab:pois_fit}). We therefore use the calibrated pairs $(\gamma^{*}, T^{*})$ throughout the dynamical simulations of Section~\ref{result:sec_results}.

\subsection{Simulation protocols}
\label{metho:subsec_protocols}
To investigate how residential mobility evolves under the GOSM, we perform large-scale Monte Carlo simulations with the pairwise swap move \eqref{eq:swap} and the Metropolis acceptance rule \eqref{eq:metropolis}. Time is measured in Monte Carlo steps, with one sweep defined as $M$ attempted swaps. All simulation routines are implemented in Python and accelerated with Numba just-in-time (JIT) compilation to handle tens of millions of trajectories efficiently.

We employ three simulation protocols.

\textbf{Ergodicity runs.}  Fixed-parameter runs of $10^{10}$ steps ($673$ sweeps) are launched from the observed, ground, and random states to test whether different initial configurations converge [Fig.~\ref{result_fig2}(a),(b)]. The observables $(\phi_{\mathrm{norm}}, E, S)$ are recorded every $10^{4}$ steps.

\textbf{Aging runs.}  The fixed-parameter simulations are extended to $10^{12}$ steps ($6.7 \times 10^{4}$ sweeps) to resolve the slow power-law relaxation and the two-time autocorrelation $C(\tau, t_w)$ over several temporal decades [Fig.~\ref{result_fig2}(c),(d)]. Results are averaged over two independent random seeds.

\textbf{Hysteresis sweeps.}  Stepped cooling and heating sweeps cover 25 logarithmically spaced temperatures $T \in [0.1, 10^{4}]$ and 25 linearly spaced values $\gamma \in [0.01, 1.2]$, with $10^{10}$ steps at each parameter value starting from the final configuration of the preceding step. The operational arrest boundary is defined against this fixed observation window, following standard practice in glass physics \cite{Angell1995,BerthierBiroli2011}.

\section{Results}
\label{result:sec_results}

\begin{figure*}[htbp]
\centering
\includegraphics[width=0.98\linewidth]{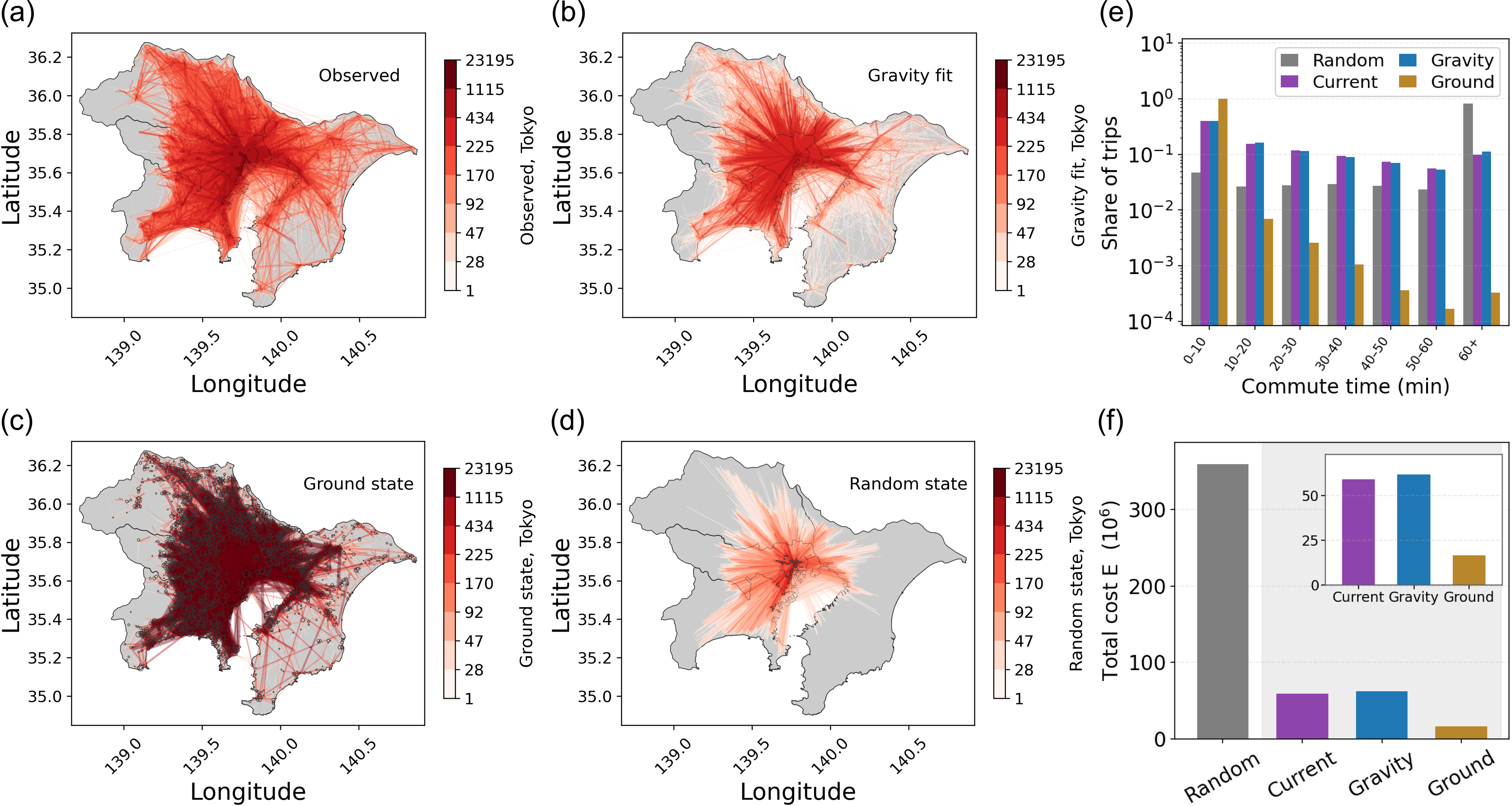}
\caption{\textbf{The observed city is nearly indistinguishable from the gravity equilibrium in trip-length distribution.}
(a)--(d)~Spatial visualization of commuting flows $f_{ij}$:
(a)~observed OD matrix;
(b)~doubly constrained gravity model with parameters from Fig.~\ref{method:fig3};
(c)~cost-minimizing ground state [Eq.~\eqref{eq:ground-def}], solved as a minimum-cost flow \cite{hitchcock1941}; circles mark zones whose commuters live and work in the same zone ($i=j$), which cannot be drawn as lines, sized by $w_j$;
(d)~random state.
Line color encodes flow volume.
(e)~Trip-length distributions (log scale): observed and gravity states nearly the same; the random state is dominated by commutes $> 60$~min; the ground state concentrates almost all trips below 10~min.
(f)~Total commuting cost $E$: the observed city lies close to the gravity fit, roughly six times below the random state, yet well above the ground state (inset).
``Gravity fit''/``Gravity'' $= \mathbf{f}^{*}$; ``Current'' $= \mathbf{f}^{\mathrm{obs}}$.}
\label{result_fig1}
\end{figure*}


Our central finding is that under the GOSM framework, the marginal-constrained Kawasaki dynamics with Metropolis swap rules arrests short of the theoretical gravity-model equilibrium. Simulations launched from different initial commuting assignments approach the theoretical equilibrium $\mathbf{f}^{*}$ at a decreasing rate and do not reach it within the simulated timescales. To establish this result, we organize our analysis around the four benchmark OD matrices mapped for Tokyo in Fig.~\ref{result_fig1}(a)--(d), which share identical housing and job capacities and differ only in how residences are matched to workplaces. All subsequent dynamical measurements track how pairwise origin swaps move the urban commuting system among these four reference states over time.


We present the non-equilibrium arrest in four steps across Sections~\ref{result:subsec_equivalence}--\ref{result:subsec_landscape}. First, Sec.~\ref{result:subsec_equivalence} establishes the theoretical baseline: a convex free-energy landscape ensures that the stationary state of the GOSM coincides with the gravity equilibrium. Second, Sec.~\ref{result:subsec_failure} shows that under the Kawasaki--Metropolis dynamics, ergodicity breaking, slow relaxation, hysteresis, and aging emerge. Third, Sec.~\ref{result:subsec_hysteresis} sweeps the behavioral parameters across the $(\gamma, T)$ plane to construct a non-equilibrium phase diagram, revealing irreversible hysteresis loops and locating the calibrated parameters of all six Japanese metropolitan areas inside the kinetically arrested regime. Finally, Sec.~\ref{result:subsec_landscape} projects the simulated relocation trajectories directly onto the computable free-energy landscape, showing how the modeled metropolitan mobility freezes near the gravitational equilibrium.

\subsection{Convexity and ensemble equivalence}
\label{result:subsec_equivalence}

Before examining the relocation trajectories, we address two theoretical questions regarding the urban free-energy landscape and the equivalence of our two modeling routes. First, does the urban free energy $\mathcal{F}(\mathbf{f})$ possess a single valley or multiple competing local valleys that could trap the city thermodynamically? Second, do Wilson's static fixed-cost formulation (microcanonical) and the GOSM's fixed-temperature dynamics (canonical) converge to the same equilibrium state in large metropolitan systems? By relaxing the integer OD matrices of Sec.~\ref{metho:subsubsec_constraints} onto the convex, compact constraint set $\Gamma$ of non-negative real matrices with fixed row sums $h_i$ and column sums $w_j$, we resolve both questions analytically below.

\subsubsection{Strict convexity of the free-energy landscape}
\label{result:subsubsec_convexity}

We prove that the urban free-energy landscape $\mathcal{F}(\mathbf{f})$ is strictly convex over $\Gamma$, guaranteeing a unique global equilibrium with no metastable traps:

\begin{proposition}[Strict convexity of $\mathcal{F}$]
\label{prop:convex}
On $\Gamma$, the Landau free energy \eqref{eq:free-energy}, written with the Stirling entropy \eqref{eq:stirling} as
\begin{equation}
  \mathcal{F}(\mathbf{f}) \;=\; \sum_{ij} f_{ij}\, t_{ij}^{\gamma} \;+\; T \sum_{ij} f_{ij}\bigl(\ln f_{ij} - 1\bigr),
  \label{eq:landau-explicit}
\end{equation}
is strictly convex and possesses a unique minimizer $\mathbf{f}^{*}$.
\end{proposition}

\noindent\textit{Proof.}
The energy term is linear in $\mathbf{f}$. The entropy term has Hessian $\partial^2 \mathcal{F}/\partial f_{ij}\,\partial f_{kl} = (T/f_{ij})\,\delta_{ik}\delta_{jl}$, diagonal with strictly positive entries on the interior of $\Gamma$, so $\nabla^2 \mathcal{F} \succ 0$ there. A strictly convex function on a convex compact set has a unique minimizer.
$\square$

In urban terms, Proposition~\ref{prop:convex} means that the city does not possess multiple competing locally optimal spatial structures [Fig.~\ref{result_fig5}(b)]. A continuous downhill path toward the gravity equilibrium $\mathbf{f}^{*}$ exists from any initial commuting pattern, so a failure of the relocation dynamics to reach $\mathbf{f}^{*}$ is not attributable to thermodynamic barriers in the landscape.

\subsubsection{Equivalence of the microcanonical and canonical descriptions}
\label{result:subsubsec_equivalence}

Next, we establish that fixing the city-wide commuting budget $E_{\mathrm{obs}}$ (Wilson's microcanonical route) and fixing temperature $T$ (the GOSM canonical route) yield identical macroscopic and microscopic stationary states when the city is large. By a general theorem of Touchette \cite{touchette2015}, the microcanonical and canonical ensembles are equivalent at all three levels of statistical description, thermodynamic, macrostate, and measure, whenever the entropy density $s(e)$ is a strictly concave function of the per-capita cost $e = E_{\mathrm{obs}}/M$, where $M$ is the total population. We now prove that this entropy density is strictly concave across its entire domain.

\begin{proposition}[Strict concavity of the entropy density]
\label{prop:concave}
The entropy density
\begin{equation}
  s(e) \;=\; \lim_{M \to \infty} \frac{1}{M} \ln \!\!\sum_{\substack{\mathbf{f} \in \Gamma \\ E[\mathbf{f}] = Me}}\!\! \Omega(\mathbf{f})
  \label{eq:entropy-density}
\end{equation}
is strictly concave on the interior of its domain.
\end{proposition}

\noindent\textit{Proof.}
At the exponential scale, the entropy density at energy $e$ is attained by the most probable macrostate on the shell:
\begin{equation}
  s(e) \;=\; \max_{\substack{\mathbf{f} \in \Gamma \\ E[\mathbf{f}] = Me}} \frac{S(\mathbf{f})}{M}.
  \label{eq:s-as-max}
\end{equation}
Fix two energies $e_1 \neq e_2$ and let $\mathbf{f}_1, \mathbf{f}_2$ be the corresponding maximizers, so that $S(\mathbf{f}_k) = M s(e_k)$ for $k = 1, 2$. For $\lambda \in (0,1)$, define the mixture
\begin{equation}
  \mathbf{f}_\lambda \;=\; (1-\lambda)\,\mathbf{f}_1 + \lambda\,\mathbf{f}_2 .
\end{equation}
We verify three properties of $\mathbf{f}_\lambda$ in turn.
(i) \emph{It satisfies the constraints.} $\Gamma$ is convex, so the convex combination of two of its elements remains in $\Gamma$.
(ii) \emph{It lies on the interpolated energy shell.} $E$ is linear in $\mathbf{f}$, so
\begin{equation}
\begin{split}
  E[\mathbf{f}_\lambda] &= (1-\lambda) E[\mathbf{f}_1] + \lambda E[\mathbf{f}_2] = M e_\lambda, \\
  e_\lambda &\equiv (1-\lambda) e_1 + \lambda e_2 .
\end{split}
\end{equation}
(iii) \emph{Its entropy exceeds the interpolated entropy.} $S$ is strictly concave on $\Gamma$ (Proposition~\ref{prop:convex}), so for $\mathbf{f}_1 \neq \mathbf{f}_2$,
\begin{equation}
  S(\mathbf{f}_\lambda) \;>\; (1-\lambda) S(\mathbf{f}_1) + \lambda S(\mathbf{f}_2) \;=\; M\bigl[(1-\lambda) s(e_1) + \lambda s(e_2)\bigr].
\end{equation}
By (i) and (ii), $\mathbf{f}_\lambda$ is an admissible candidate in the maximization \eqref{eq:s-as-max} at energy $e_\lambda$, so $s(e_\lambda) \geq S(\mathbf{f}_\lambda)/M$. Combining with (iii),
\begin{equation}
  s\bigl((1-\lambda) e_1 + \lambda e_2\bigr) \;>\; (1-\lambda)\, s(e_1) + \lambda\, s(e_2),
\end{equation}
which is strict concavity.
$\square$

Proposition~\ref{prop:concave} states that averaging two distinct OD matrices entry-by-entry preserves the housing and job capacities, averages their commuting costs, and strictly increases their configurational diversity $S$. Consequently, the entropy density $s(e)$ is smoothly dome-shaped everywhere, with no flat linear segments or inward bulges. Because the inverse temperature $\beta = s'(e)$ is the slope of this dome, every per-capita commuting cost $e$ maps uniquely to a single temperature $T = 1/\beta$, so that Wilson's Lagrange multiplier and the GOSM's temperature coincide in the large-$M$ limit. Moreover, the absence of any flat segment in $s(e)$ rules out any first-order phase transition in the static equilibrium model.

This strict concavity yields two practical quantitative consequences for large metropolitan areas (derived via saddle-point approximation in Appendix~\ref{app:equiv}). First, relative fluctuations in the aggregate commuting cost under the GOSM vanish inversely with the square root of the commuter population,
\begin{equation}
  \frac{\sqrt{\langle \Delta E^2\rangle}}{\langle E \rangle} \;\sim\; \frac{1}{\sqrt{M}} \;\xrightarrow[M\to\infty]{}\; 0,
  \label{eq:fluctuations}
\end{equation}
meaning that at metropolitan scales ($M \sim 10^{7}$), a model city regulated by a temperature $T$ is numerically indistinguishable from one with a rigidly fixed total commuting time. Second, the mean commuting cost $\langle E \rangle(T)$ is a smooth, strictly increasing function of $T$, so any admissible commuting budget (the energy $E$ in this study) is realized by a unique temperature. Together, Propositions~\ref{prop:convex} and~\ref{prop:concave} show that the equilibrium landscape contains no local minima, so obstruction to reaching $\mathbf{f}^{*}$ arises from the kinetic constraints of pairwise residential swaps.

\subsection{Failure of equilibration}
\label{result:subsec_failure}

\begin{figure}[h]
\centering
\includegraphics[width=0.95\linewidth]{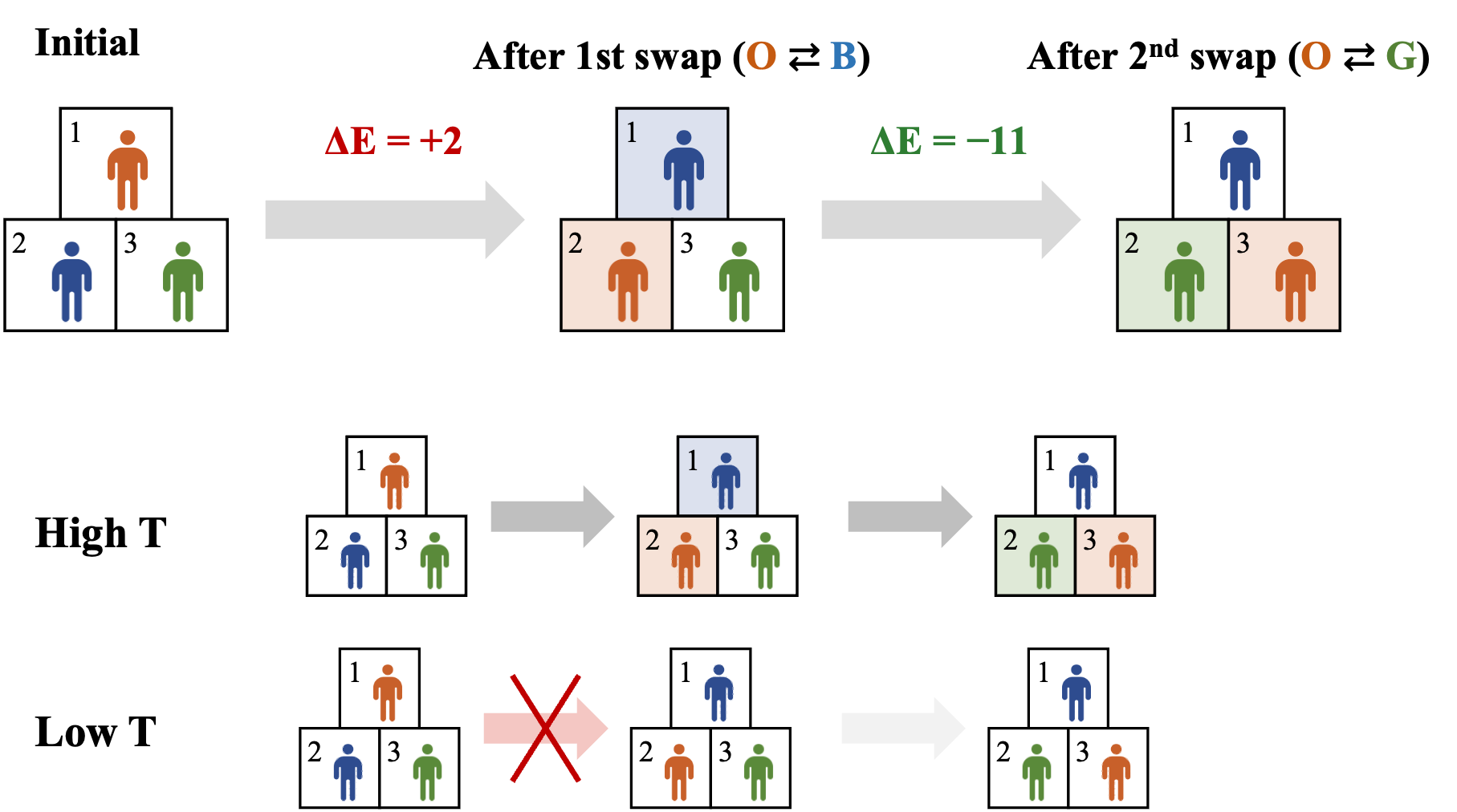}
\caption{\textbf{Reaching the global optimum requires an uphill intermediate swap that low-temperature dynamics suppresses.}
Three commuters occupy three homes with fixed workplaces.
The globally optimal assignment requires two sequential pairwise swaps; the first temporarily increases total commuting cost ($\Delta E > 0$) before the second reduces it by a larger amount (illustrative values shown).
At high temperature $T$ the uphill swap is readily accepted and the system equilibrates; at low $T$ the Metropolis rule exponentially suppresses it, trapping the commuters in a suboptimal assignment even though a lower-cost configuration is only two swaps away.}
\label{fig:arrest}
\end{figure}

\begin{figure*}[!t]
\centering
\includegraphics[width=0.8\linewidth]{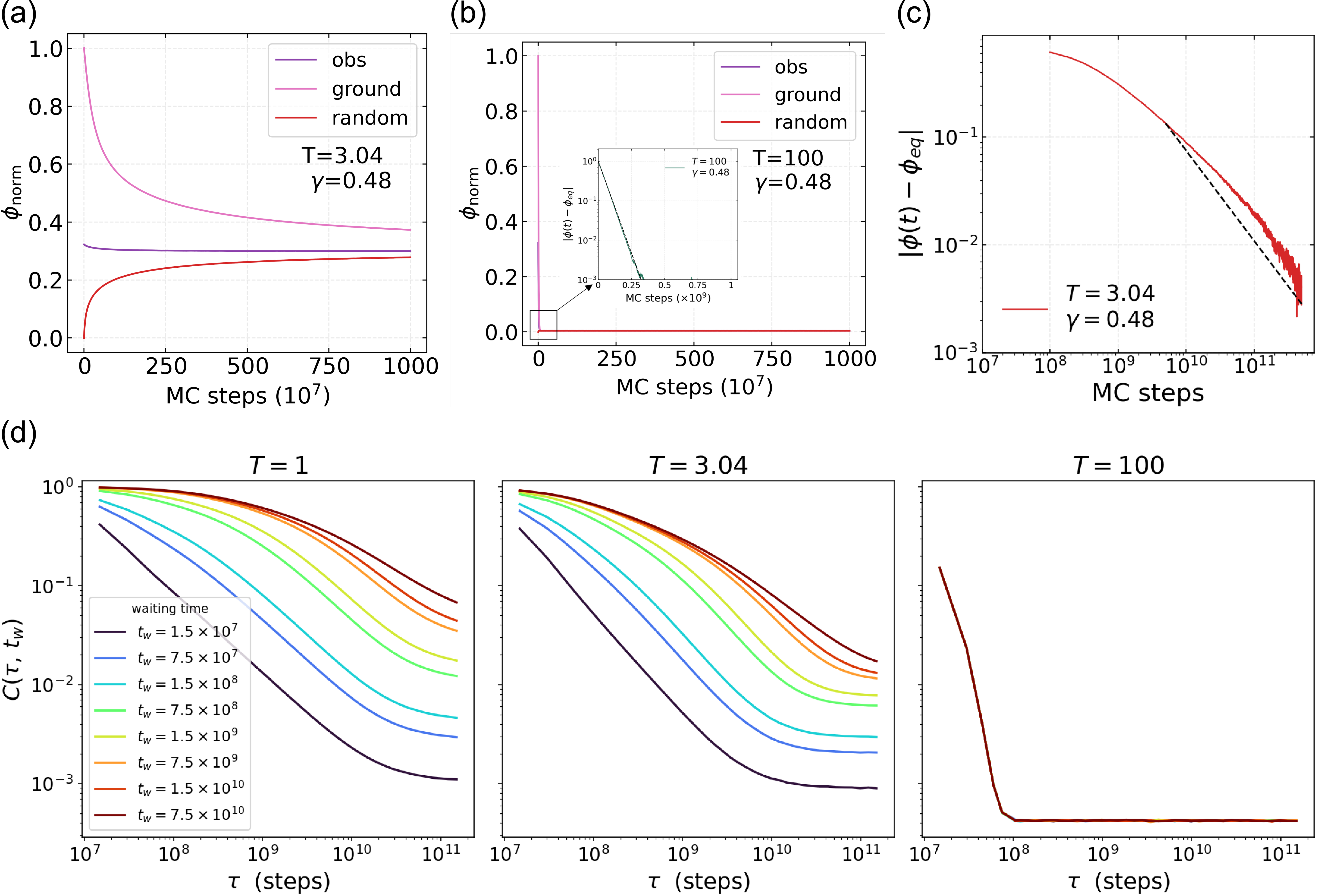}
\caption{\textbf{Relocation process fails to equilibrate: trajectories retain initial-condition memory, relax as a power law, and together with aging.}
(a)~Order parameter $\phi_{\mathrm{norm}}$ versus Monte Carlo steps from three initial states at Tokyo's calibrated parameters ($T = 3.04$, $\gamma = 0.48$).
The trajectories flatten without merging over $10^{10}$ steps, indicating broken ergodicity.
(b)~At high temperature ($T = 100$), all three trajectories collapse onto a single equilibrium value; the inset confirms exponential relaxation.
(c)~At the calibrated parameters, an extended run ($10^{12}$ steps) reveals power-law decay (dashed line: guide line).
(d)~Two-time autocorrelation $C(\tau, t_w)$ [Eq.~\eqref{eq:two-time}] at $T = 1$, $3.04$, and $100$.
At $T = 1$ and $3.04$ the decay shifts to longer lag times as $t_w$ increases (aging), whereas at $T = 100$ all curves collapse (stationarity).}
\label{result_fig2}
\end{figure*}

We find that the pairwise origin-swapping dynamics fails to reach gravity equilibrium within a limited time at the empirically calibrated city parameters, because every housing unit is occupied and improvements are realized only through the acceptance of individual swaps. Fig.~\ref{fig:arrest} shows the mechanism with a minimal three-commuter example: rearranging three or more commuters into a better overall configuration requires a sequence of swaps, and the first swap typically forces one pair into a temporarily higher combined commuting cost ($\Delta E > 0$). At high temperature ($T=100$) such uphill intermediate swaps are easily accepted and the system quickly finds the global minimum, whereas at low $T$ the Metropolis rule rarely accepts energy-increasing swaps, so residents stay in place and the system evolves slowly. In real cities, households tolerate their current commute up to a threshold rather than relocating for marginal gains, and when they do move they usually shorten the commute \cite{clark2003}. The swap rule reproduces this behavior: improvements that require someone to first accept a longer commute are rarely initiated, so the system arrests before reaching the cost-minimizing assignment.

At the macroscopic scale, this microscopic mechanism manifests first as ergodicity breaking and scale-free power-law relaxation [Fig.~\ref{result_fig2}(a)--(c)]. When simulated at Tokyo's calibrated parameters ($T^{*}=3.04$, $\gamma^{*}=0.48$), relocation trajectories initialized from the observed, ground, and random states fail to merge even after $10^{10}$ Monte Carlo steps, flattening at distinct values of $\phi_{\mathrm{norm}}$ and retaining a memory of their initial residential patterns [Fig.~\ref{result_fig2}(a); additional parameter points are shown in SM Sec.~S2]. By contrast, raising the temperature to $T=100$ allows all three initializations to collapse almost immediately onto the same equilibrium state via conventional exponential relaxation [Fig.~\ref{result_fig2}(b) and inset]. At the calibrated city parameters, extending the simulation to $10^{12}$ steps reveals that the approach toward stationarity crosses over from exponential decay to a slow power law spanning more than three decades in time [Fig.~\ref{result_fig2}(c)], indicating relaxation without a finite characteristic timescale.

To test whether individual residential mobility also stiffens over time, a phenomenon known in glass physics as \emph{physical aging}?we measure the two-time residence-persistence autocorrelation function [Fig.~\ref{result_fig2}(d)]. Formally, we define $C(\tau, t_w)$ as the fraction of commuters who remain in the exact same residential grid cell between a waiting time $t_w$ (measured from the start of the simulation) and a subsequent lag time $t_w + \tau$:
\begin{equation}
  C(\tau, t_w)
  \;=\;
  \Bigl\langle
    \frac{1}{M} \sum_{c=1}^{M}
    \delta\bigl( o_c(t_w + \tau),\, o_c(t_w) \bigr)
  \Bigr\rangle ,
  \label{eq:two-time}
\end{equation}
where $o_c(t)$ denotes the residential zone of commuter $c$ at step $t$. In an equilibrated system, residential turnover is stationary (time-translation invariant), meaning $C(\tau, t_w)$ depends only on the elapsed lag $\tau$ and not on the waiting time $t_w$, as observed at high temperature $T=100$ where all curves collapse onto a single fast decay [Fig.~\ref{result_fig2}(d), right]. By contrast, at $T=1$ and at Tokyo's calibrated temperature $T^{*}=3.04$, the persistence curves shift systematically to longer lag times $\tau$ as the waiting time $t_w$ increases [Fig.~\ref{result_fig2}(d), left and center]. In urban terms, as easily accessible, mutually beneficial origin swaps are progressively exhausted, the remaining residential configuration becomes increasingly rigid, the longer the system has evolved, the slower residents relocate, mirroring the aging behavior recently identified in stylized Schelling segregation models \cite{abella2022}.

Together, the three dynamical measurements in Fig.~\ref{result_fig2} indicate a breakdown of equilibration within the calibrated GOSM framework, with each panel negating a requirement of thermal equilibrium. First, equilibrium requires independence from initial conditions, yet Fig.~\ref{result_fig2}(a) shows broken ergodicity across different starting assignments. Second, equilibrium requires relaxation within a finite characteristic timescale, yet Fig.~\ref{result_fig2}(c) shows scale-free power-law relaxation that does not complete within the accessible observation window. Third, equilibrium requires stationary turnover rates, yet Fig.~\ref{result_fig2}(d) shows physical aging where relocation slows down as the configuration matures. Having shown that the commuting state at fixed $(\gamma^{*}, T^{*})$ depends on both its initial condition and its elapsed age, we next investigate whether it also depends on the historical path along which the behavioral parameters were varied.

\subsection{Hysteresis and the nonequilibrium phase diagram}
\label{result:subsec_hysteresis}

\begin{figure*}[!t]
\centering
\includegraphics[width=0.98\linewidth]{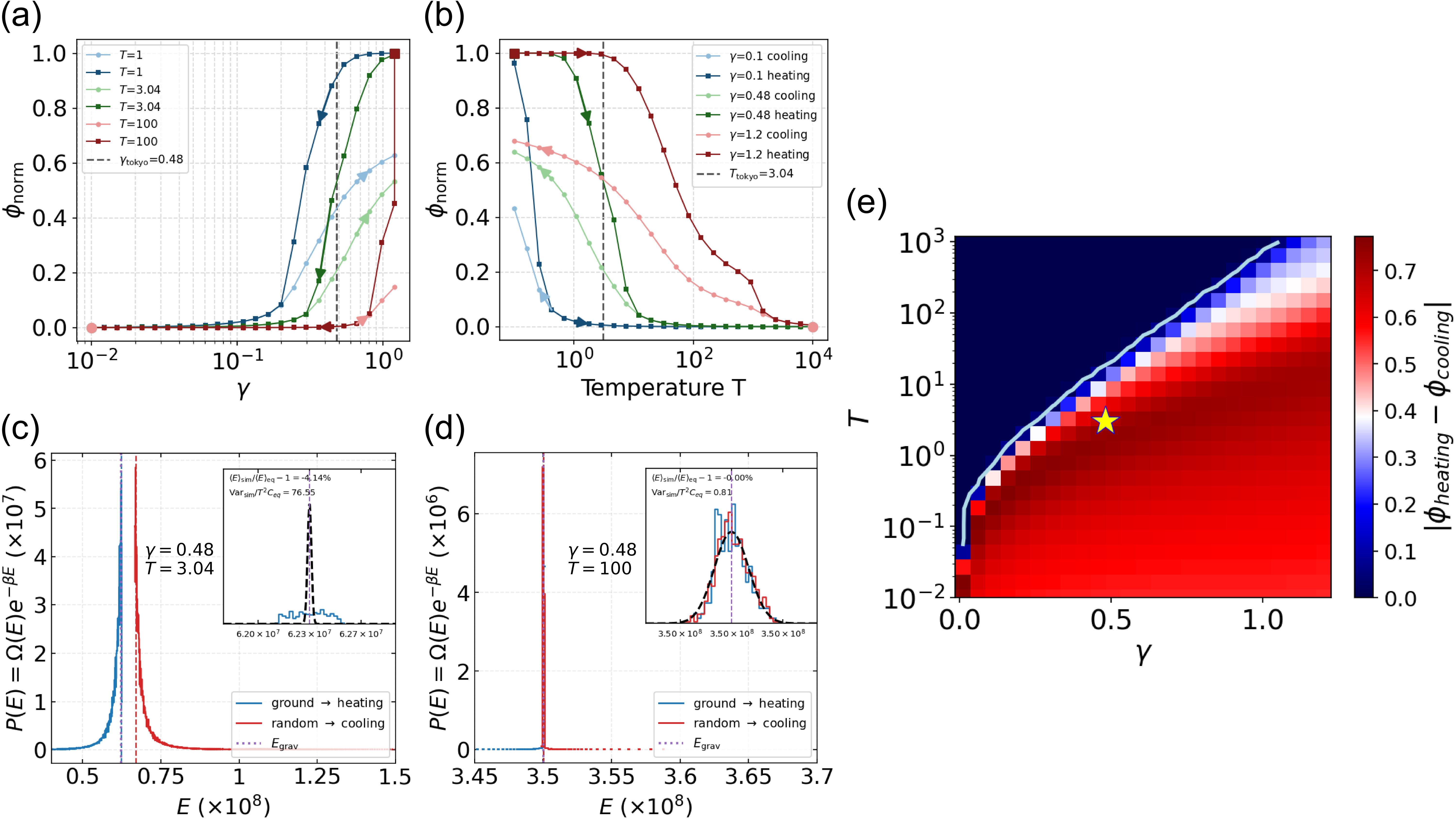}
\caption{\textbf{All six Japanese metropolitan areas fall inside the kinetically arrested regime, where commuting patterns depend on history rather than equilibrium (hysteresis).} (a)~Order parameter $\phi_{\mathrm{norm}}$ under stepped quasistatic sweeps of the perception exponent $\gamma$ at fixed $T=1$, $3.04$, and $100$ ($10^{10}$ Monte Carlo steps per point; circles and forward arrows denote increasing $\gamma$, squares and backward arrows denote decreasing $\gamma$). At $T=1$ and $T=3.04$, the forward and backward trajectories enclose an irreversible hysteresis loop, whereas at $T=100$ the loop closes and the sweep is reversible. The vertical dashed line marks Tokyo's calibrated $\gamma^{*}=0.48$, whose intersection with the $T=3.04$ branches (green) brackets the accessible range of $\phi_{\mathrm{norm}}$. (b)~Corresponding hysteresis sweeps across temperature $T$ at fixed $\gamma=0.1$, $0.48$, and $1.2$ (vertical dashed line at $T^{*}=3.04$), showing that cooling from the random state fails to recover the compact job--housing order destroyed during heating. (c)~Sampled commuting-cost distributions $P(E)=\Omega(E)e^{-\beta E}$ at Tokyo's calibrated parameters ($T=3.04$, $\gamma=0.48$) obtained after heating from the ground state (blue) versus cooling from the random state (red). The two distributions fail to overlap, and the sampled variance exceeds the equilibrium fluctuation--dissipation prediction $\mathrm{Var}_{\mathrm{sim}}/(T^{2}C_{\mathrm{eq}}) = 1$ by nearly two orders of magnitude (inset). (d)~At $T=100$, the heating and cooling distributions collapse onto each other and match the canonical energy distribution, which is Gaussian in the large system (dashed curve), with $\langle E\rangle_{\mathrm{sim}}/\langle E\rangle_{\mathrm{eq}}-1\approx 0$ and $\mathrm{Var}_{\mathrm{sim}}/(T^{2}C_{\mathrm{eq}})\approx 0.8$. (e)~Equilibrium phase diagram across the $(\gamma, T)$ plane: the residual hysteresis $|\phi_{\mathrm{heating}}-\phi_{\mathrm{cooling}}|$ separates an ergodic regime (dark blue, upper left) from a kinetically arrested regime (red, lower right) across the $|\phi_{\mathrm{heating}}-\phi_{\mathrm{cooling}}|=0.05$ operational boundary (light contour), with Tokyo's calibrated point (yellow star) located inside the arrested region.}
\label{result_fig4}
\end{figure*}

To test whether the commuting configuration depends on the historical path of parameter changes, we subject the GOSM to stepped cooling and heating sweeps across $\gamma$ and $T$, uncovering pronounced hysteresis loops [Fig.~\ref{result_fig4}(a),(b)] \cite{Angell1995,BerthierBiroli2011}. In an urban context, cooling ($T$ decreasing or $\gamma$ increasing) represents a progressive tightening of commuters' sensitivity to travel time, whereas heating represents a relaxation toward indifference. At high temperature ($T=100$) the forward and reverse $\gamma$-sweeps coincide, showing reversible spatial sorting [Fig.~\ref{result_fig4}(a)]. At $T=1$ and at Tokyo's calibrated $T^{*}=3.04$, the sweeps split into an open hysteresis loop; the same irreversibility appears when sweeping $T$ at fixed $\gamma$ [Fig.~\ref{result_fig4}(b)]. Tokyo's calibrated parameters intersect the steep flanks of these loops, indicating that the model city sits in a history-dependent regime where small shifts in travel-cost sensitivity alter achievable commuting patterns.

This path dependence is confirmed at the distribution level by comparing the sampled commuting-cost histograms against the canonical energy distribution $P(E) \propto \Omega(E)\,e^{-\beta E}$ [Fig.~\ref{result_fig4}(c),(d)]. In canonical equilibrium, the sampled cost distribution converges to a single Gaussian peak with mean $\langle E\rangle_{\mathrm{eq}}$ and variance governed by the fluctuation--dissipation relation $\mathrm{Var}(E)=T^{2}C_{\mathrm{eq}}$. At the calibrated Tokyo parameters [Fig.~\ref{result_fig4}(c)], the cost histograms sampled after heating from the ground state (blue) and cooling from the random state (red) remain separated, and the sampled variance exceeds the equilibrium fluctuation--dissipation prediction by nearly two orders of magnitude (inset). Only when the temperature is raised to $T=100$ [Fig.~\ref{result_fig4}(d)] do the heating and cooling histograms collapse onto the single Gaussian peak predicted by the canonical ensemble, with $\langle E\rangle_{\mathrm{sim}}/\langle E\rangle_{\mathrm{eq}}-1\approx 0$, indicating that equilibration within the observation window occurs at high temperature.

Mapping the residual hysteresis $|\phi_{\mathrm{heating}}-\phi_{\mathrm{cooling}}|$ across the entire $(\gamma, T)$ plane yields a glass-like non-equilibrium phase diagram that places all six empirical metropolitan calibrations inside the kinetically arrested regime [Fig.~\ref{result_fig4}(e)]. Using the $|\phi_{\mathrm{heating}}-\phi_{\mathrm{cooling}}| = 0.05$ contour as an operational boundary separates the parameter space into a reversible ergodic regime at high $T$ and low $\gamma$ (dark blue, upper left) and a history-dependent arrested regime at low $T$ and high $\gamma$ (red, lower right). Tokyo's calibrated parameter pair $(\gamma^{*}=0.48, T^{*}=3.04$, yellow star) lies inside the arrested region just below the crossover boundary, and repeating this procedure across Osaka, Nagoya, Fukuoka, Sapporo, and Sendai reveals that all six metropolitan areas fall within the same arrested regime (Appendix~\ref{app:sixcities}, Fig.~\ref{app_fig_arrest}). To visualize \emph{how} the relocation dynamics stalls near this boundary, we next project the simulated trajectories directly onto the urban free-energy landscape.

\subsection{Trajectories in the energy and free-energy landscapes}
\label{result:subsec_landscape}

\begin{figure}[htbp]
\centering
\includegraphics[width=0.65\linewidth]{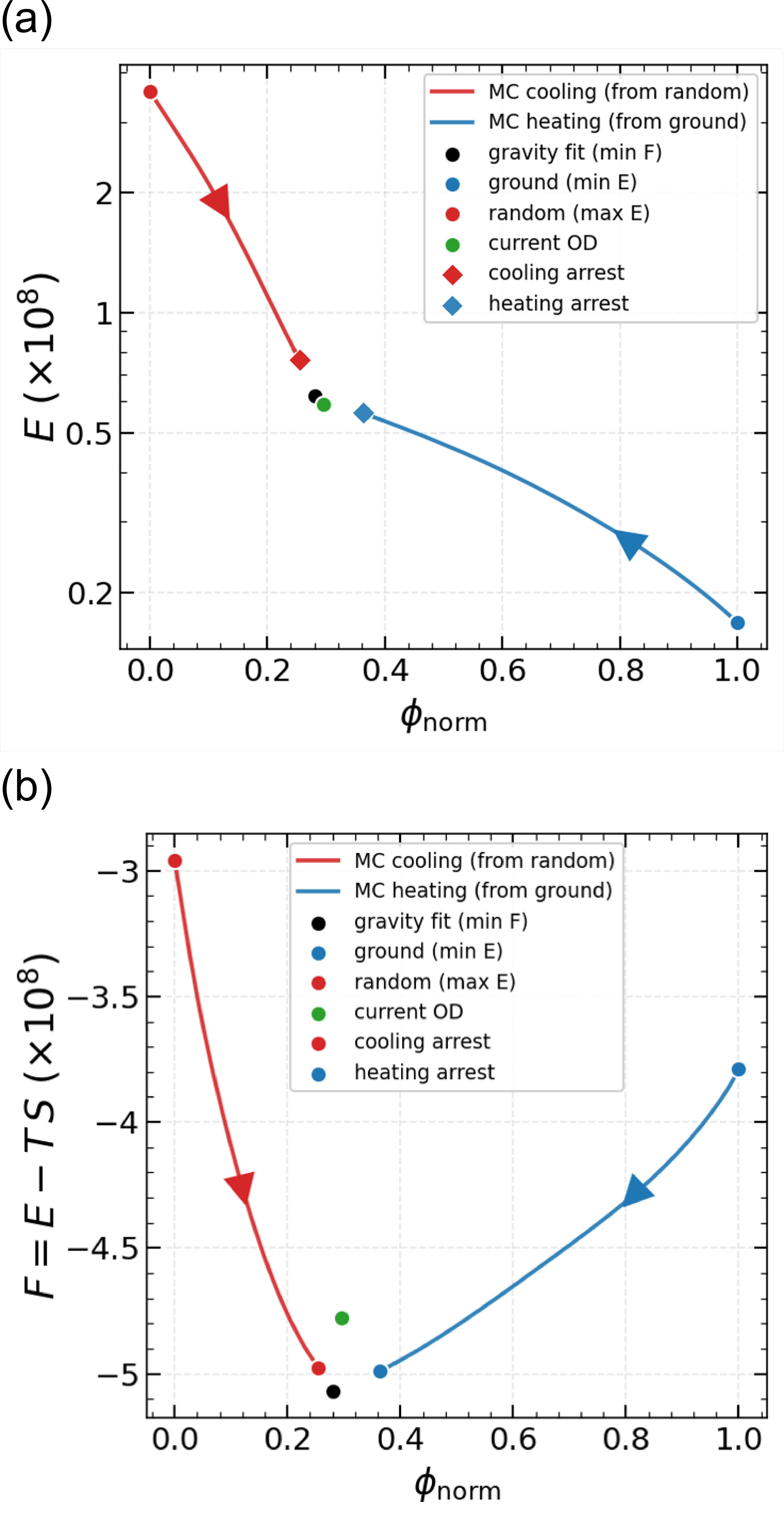}
\caption{\textbf{Energy and free-energy landscapes: simulated trajectories freeze on both sides of the free-energy minimum, consistent with kinetic rather than thermodynamic arrest.} (a)~Aggregate commuting cost $E$ versus the job--housing proximity order parameter $\phi_{\mathrm{norm}}$ along Monte Carlo trajectories. The cooling trajectory initialized from the random state (red, starting at $\phi_{\mathrm{norm}}=0$) and the heating trajectory initialized from the ground state (pink, starting at $\phi_{\mathrm{norm}}=1$) both move toward the theoretical gravity equilibrium (black circle) and observed city (green circle), but stall at distinct arrest points (red and blue diamonds) before meeting. (b)~The identical trajectories plotted on the computable urban free-energy landscape $\mathcal{F} = E - TS$. Consistent with the strict convexity proved in Proposition~\ref{prop:convex}, both trajectories descend monotonically along the smooth walls of a single convex valley toward the unique global minimum $\mathbf{f}^{*}$ (black circle), yet each freezes on its own side of the minimum, leaving an unclosed free-energy gap. Similarly, a trajectory initialized directly from the observed OD matrix relaxes only marginally before stalling, consistent with kinetic arrest on a convex free-energy landscape.}
\label{result_fig5}
\end{figure}

The kinetic arrest can be visualized by plotting the step-by-step simulation trajectories on the energy and free-energy landscapes (Fig.~\ref{result_fig5}), since both the aggregate commuting cost $E[\mathbf{f}]$ and the configurational entropy $S(\mathbf{f})$ can be evaluated exactly for any OD matrix $\mathbf{f}$ (Sec.~\ref{metho:subsubsec_landau}). In the commuting-cost plane [Fig.~\ref{result_fig5}(a)], the random state sits at $\phi_{\mathrm{norm}}\approx 0$ with maximal travel cost and the ground state sits at $\phi_{\mathrm{norm}}=1$ with minimal travel cost. Over $10^{10}$ Monte Carlo steps, both the cooling trajectory launched from the random state (red) and the heating trajectory launched from the ground state (pink) move toward the middle of the plot but stop at distinct arrest points on either side of the theoretical equilibrium $\mathbf{f}^{*}$. While the observed OD matrix $\mathbf{f}^{\mathrm{obs}}$ (green circle) and the gravity equilibrium $\mathbf{f}^{*}$ (black circle) share nearly the same aggregate cost $E$ in Fig.~\ref{result_fig5}(a), this overlap reflects the trip-length calibration of the gravity model rather than equilibration.

Replotting these trajectories on the urban free-energy landscape $\mathcal{F}(\mathbf{f}) = E[\mathbf{f}] - T^{*} S(\mathbf{f})$ resolves this cost degeneracy and shows the entropic gap separating the arrested trajectories from the global equilibrium [Fig.~\ref{result_fig5}(b)]. Although the observed city $\mathbf{f}^{\mathrm{obs}}$ and the theoretical equilibrium $\mathbf{f}^{*}$ have nearly identical total commuting costs $E$, the observed city has a lower configurational entropy $S(\mathbf{f})$ and therefore sits at a higher free energy $\mathcal{F}$ above the global minimum $\mathbf{f}^{*}$. Along the simulation runs, the cooling branch (red), heating branch (pink), and a trajectory launched directly from the observed city (purple) all descend $\mathcal{F}$ monotonically along the smooth walls of the single convex valley predicted by Proposition~\ref{prop:convex}. Yet despite the absence of local free-energy barriers or metastable traps, all three trajectories stop short of the minimum, leaving a persistent free-energy gap. Figure~\ref{result_fig5}(b) thus indicates that the failure to reach the gravity equilibrium is kinetic rather than thermodynamic.

\section{Discussion}
\label{sec:discussion}


This study has given the doubly constrained gravity model a dynamical counterpart, the GOSM, in order to describe how the gravity equilibrium would be reached through residential relocation. The model combines two ingredients: Kawasaki-type origin swaps, which preserve the housing and job capacities of every zone, and a Metropolis acceptance rule, which accepts a longer commute with a probability set by the temperature $T$; both parameters are calibrated from observed commuting data. Our main finding is that, although the free-energy landscape of this model is strictly convex and its equilibrium is the gravity state $\mathbf{f}^{*}$, the relocation dynamics at the calibrated parameters does not reach it. Instead it shows the non-equilibrium phenomena of glassy dynamics: ergodicity breaking, hysteresis, and aging, which together constitute kinetic arrest.

The GOSM is a simplified model of the real relocation process. Real cities differ from it in two opposite ways. On the one hand, households face extra constraints such as job changes, childcare, moving costs, and homeownership, which make relocation harder than in the model. On the other hand, real cities have vacant homes, which allow a household to move without finding a swap partner and thus open relocation paths that pairwise swaps cannot reach. Because these two effects pull in opposite directions, the model alone cannot tell us whether real cities are closer to or further from the gravity equilibrium than the simulation suggests. A direct validation requires longitudinal residence data, which the present GPS data do not provide. The most immediate extension is to introduce a vacancy rate $\rho$, the fraction of vacant units in the existing stock: vacancies supply intermediate states that pairwise swaps cannot, so we expect a crossover from arrested to equilibrating dynamics once $\rho$ exceeds a threshold \cite{kob1993,RitortSollich2003}. This extension would turn the present result into a testable statement about real cities.

The mechanism has an urban reading. The surviving relaxation near equilibrium is carried by short-distance commuters, and the accepted rearrangements become increasingly short-range: what separates the model city from its equilibrium is a large number of micro-adjustments among nearby residents. For the Metropolis algorithm these are low-gain proposals; for households they are relocations worth a few minutes of travel time, below typical moving costs. The arrest reflects this gap between algorithmic and behavioral optimization. Policymakers should interpret static OD snapshots with caution: arrested dynamics means that observed flows may underestimate the welfare gains achievable through structural intervention. Substantial gains may require changing the landscape itself, through transport investment that reshapes $t_{ij}$ or housing development that reshapes the marginals.

Two further extensions are speculative. First, frictions beyond travel time may deepen the arrest: if short-distance commuters disproportionately face non-commuting attachments such as tenure, ownership, or schooling, the effective barriers for exactly the moves that remain are higher still; testing this requires linking OD data to tenure records. Second, because the arrest is kinetic, it could be circumvented by changing the move set rather than the incentives: coordinated exchange mechanisms that execute long chains of relocations as a single transaction, in the spirit of clearinghouse designs for kidney exchange \cite{abdulkadiroglu1999}, correspond to non-local moves that bypass the barriers responsible for freezing. The same construction can be applied wherever a maximum-entropy model is fitted to a static snapshot: adding constraint-preserving moves turns the equilibrium assumption into a testable claim rather than a premise.

\section*{Author Contributions}
Y.Y.Z. designed the research plan, developed the data analysis methods, performed the numerical calculations, and wrote the manuscript. H.T. designed the research plan, verified the data analysis methods, and revised the manuscript. M.T. led the project and directed the writing of the manuscript.

\begin{acknowledgments}
This work was supported by the Japan Society for the Promotion of Science, Grant-in-Aid for Scientific Research (B) (GrantNumber 23K22980 to MT). The funder had no role in study design, data collection and analysis, decision to publish, or preparation of the manuscript.
\end{acknowledgments}

\section*{Data Availability}
The GPS data cannot be shared publicly because they are available only on request from a third party. Data are available from Agoop Corporation, a Japanese private company that provides location big data acquired from smartphone applications. The specific product is ``Pointo-gata ryoudou-jinkou data'' (point-type population data). Researchers who meet the criteria for access to confidential data can visit \url{https://agoop.co.jp/service/dynamic-population-data/} for more information.

Zone definitions use the 1~km national grid squares from the National Land Numerical Information dataset provided by MLIT, freely available under the CC BY 4.0 license \cite{mlit_mesh1000}.

\appendix

\section{Data construction}
\label{app:data}

\begin{figure}[htbp]
\centering
\includegraphics[width=0.85\linewidth]{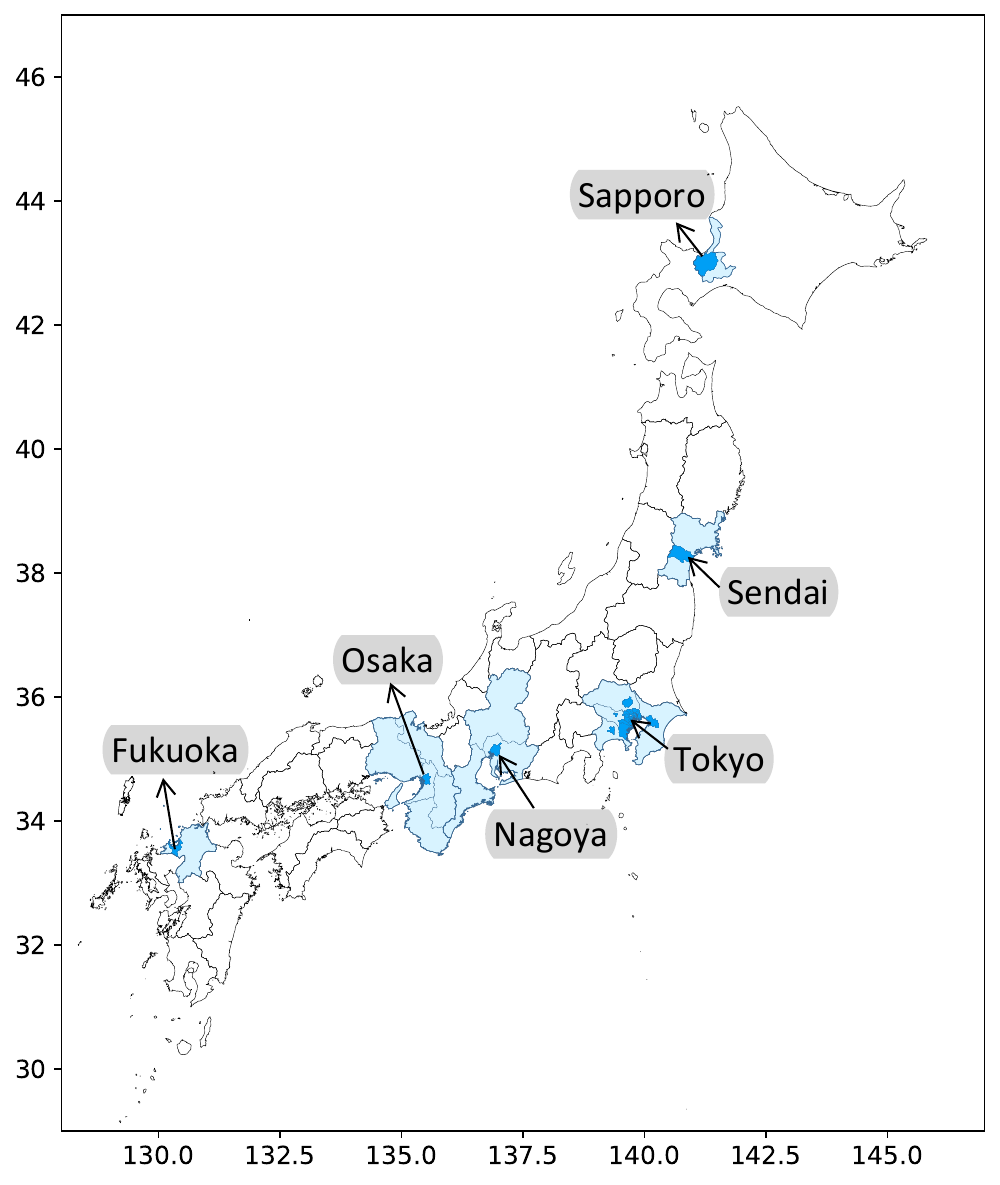}
\caption{The six metropolitan areas analyzed in this study---Sapporo, Sendai, Tokyo, Nagoya, Osaka, and Fukuoka---each with a population above 2~million \cite{estat2024}.}
\label{app_fig_map}
\end{figure}

Each GPS record carries a randomized user ID (reset nightly), a timestamp, coordinates (one-minute resolution, spatial accuracy about 10~m), and home and work city codes. Users with fewer than 100 location points per day are excluded; holidays are excluded. Each point is labeled \emph{home}, \emph{work}, \emph{move}, \emph{stroll}, or \emph{stay} by thresholds on speed, dwell time, and grid residence \cite{zheng2024plosone}: \emph{home} is the 100~m cell of the longest nighttime stay (dwell over 4~h), \emph{work} the daytime cell with dwell over 5~h distinct from home. The morning commute runs from the first non-home status after home to the first work status, giving the door-to-door time $T_{\mathrm{commute}}$; $t_{ij}$ is its median over trajectories on the pair $(i,j)$. Known limitations (demographic bias of smartphone users; nightly ID reset) are discussed in Ref.~\cite{zheng2024plosone}. Figure~\ref{app_fig_map} shows the locations of the six metropolitan areas, and Table~\ref{app_tab:data} summarizes the resulting OD data.

Pairs with no observed trips have no measured $t_{ij}$; we assign them a penalty time of 1200~min, since IPF requires a strictly positive kernel. The penalty is large enough that the thermal equilibrium puts only $\sim 10^{-7}$ of the $\sim 1.5 \times 10^{7}$ trajectories on these pairs, so its exact value affects no calibrated result---only the energy of the random state, which places flow there regardless of travel time.

\begin{table}[htbp]
  \caption{\label{app_tab:data}Summary of commuting OD data by metropolitan area.
  $N$ is the number of spatial units (grid cells), i.e., the linear dimension
  of the $N \times N$ OD matrix; ``Nonzero OD pairs'' is the number of matrix
  elements with $f_{ij} > 0$; and $M = \sum_{ij} f_{ij}$ is the total number
  of commuting trajectories over the year, approximately 30 million in total.}
  \begin{ruledtabular}
  \begin{tabular}{lccc}
  Metropolitan area & $N$ & Nonzero OD pairs & $M$ \\
  \hline
  Tokyo   & 9{,}750 & 517{,}196 & 14{,}863{,}988 \\
  Osaka   & 9{,}675 & 274{,}195 & 7{,}014{,}248 \\
  Nagoya  & 8{,}787 & 192{,}189 & 4{,}290{,}435 \\
  Fukuoka & 3{,}018 & 59{,}936 & 1{,}548{,}570 \\
  Sapporo & 8{,}567 & 64{,}249 & 1{,}293{,}554 \\
  Sendai  & 2{,}651 & 26{,}472 & 663{,}977 \\
  \end{tabular}
  \end{ruledtabular}
\end{table}

\section{Derivations for the gravity model and the calibration}
\label{app:derivations}
\subsection{Limits of the stretched-exponential deterrence}
\label{app:limits}

The stretched-exponential deterrence kernel $F(t) = e^{-\beta t^{\gamma}}$ in Eq.~\eqref{eq:deterrence} unifies the classical exponential and power-law forms through a Box--Cox transformation of travel time. Specifically, defining the transformed travel cost $c_{\gamma}(t)$ and rescaled sensitivity $\tilde{\beta} = \beta\gamma$ yields
\begin{equation}
  c_{\gamma}(t) = \frac{t^{\gamma}-1}{\gamma},
  \qquad
  F(t) = e^{-\beta t^{\gamma}} = e^{-\beta}\,e^{-\tilde{\beta}\,c_{\gamma}(t)},
  \label{eq:boxcox}
\end{equation}
where the constant prefactor $e^{-\beta}$ is absorbed directly into the regional balancing factors $A_i$ and $B_j$ in Eq.~\eqref{eq:gravity}. Evaluating Eq.~\eqref{eq:boxcox} at the two limiting exponents gives
\begin{align}
  \gamma = 1: &\quad c_{1}(t) = t - 1, \qquad F(t) \propto e^{-\tilde{\beta} t},
  \label{eq:explimit}\\
  \gamma \to 0: &\quad c_{\gamma}(t) \to \ln t, \qquad F(t) \propto e^{-\tilde{\beta}\ln t} = t^{-\tilde{\beta}},
  \label{eq:powerlimit}
\end{align}
which recovers the simple exponential deterrence at $\gamma = 1$ and a pure power-law deterrence $F(t) \propto t^{-\alpha}$ (with finite exponent $\alpha = \tilde{\beta} = \beta\gamma$) in the limit $\gamma \to 0$ at fixed $\tilde{\beta}$, which in the original variables corresponds to $\beta \to \infty$. For intermediate values $0 < \gamma < 1$, the stretched exponential interpolates smoothly between these two classical regimes, behaving like a power law at moderate travel times while remaining exponentially bounded in the long-time tail.

\subsection{Maximum-entropy derivation of the doubly constrained gravity model}
\label{app:maxent}
Wilson's derivation identifies the equilibrium OD matrix as the macrostate $\mathbf{f}$ realized by the largest number of individual commuter assignments \cite{wilson1967,wilson1970}. Because $M$ distinguishable commuters can be assigned to the cell counts $\{f_{ij}\}$ in
\begin{equation}
  \Omega(\mathbf{f}) \;=\; \frac{M!}{\prod_{ij} f_{ij}!}
  \label{eq:multiplicity_app}
\end{equation}
distinct ways, the configurational entropy is given by Stirling's approximation \cite{boltzmann1877,jaynes1957} as
\begin{equation}
  S(\mathbf{f}) \;=\; \ln \Omega(\mathbf{f}) \;\simeq\; -\sum_{ij} f_{ij} \ln f_{ij} + \text{const},
  \label{eq:stirling}
\end{equation}
where the constant $M \ln M - M$ is fixed by $\sum_{ij} f_{ij} = M$. Maximizing $S(\mathbf{f})$ subject to the housing and job capacities \eqref{eq:marginals} and a fixed city-wide commuting time budget $E[\mathbf{f}]$ (with Lagrange multipliers $\lambda_i$, $\mu_j$, and $\beta$) yields the stationarity condition
\begin{equation}
  -\ln f_{ij} - 1 - \lambda_i - \mu_j - \beta t_{ij}^{\gamma} = 0,
  \label{eq:stationarity-wilson}
\end{equation}
whose unique solution is the doubly constrained gravity model of Eq.~\eqref{eq:gravity}, with regional balancing factors defined self-consistently by
\begin{equation}
  A_i = \Bigl[\sum_j B_j w_j\, e^{-\beta t_{ij}^{\gamma}}\Bigr]^{-1}, \qquad B_j = \Bigl[\sum_i A_i h_i\, e^{-\beta t_{ij}^{\gamma}}\Bigr]^{-1},
  \label{eq:balancing}
\end{equation}
which we solve via iterative proportional fitting (IPF) \cite{furness1965,deming1940,sinkhorn1967}.

\subsection{Trip-length distribution calibration objective}
\label{app:tld}
To calibrate $(\gamma, T)$ in Sec.~\ref{metho:subsubsec_heatbath}, commuting times are grouped into 5-minute intervals $b = 1, \dots, B$, where $p_b[\mathbf{f}]$ denotes the share of trips falling in bin $b$:
\begin{equation}
  p_b[\mathbf{f}] \;=\; \frac{1}{M} \sum_{(i,j)\,:\, t_{ij} \in b} f_{ij}, \qquad \sum_{b=1}^{B} p_b[\mathbf{f}] = 1.
  \label{eq:tld-def}
\end{equation}
We minimize the total-variation distance between the observed and gravity-model trip-length distributions,
\begin{equation}
  D_{\mathrm{tld}}(\gamma, T) \;=\; \frac{1}{2} \sum_{b=1}^{B} \bigl|\, p^{\mathrm{obs}}_b - p^{\mathrm{grav}}_b(\gamma, T) \,\bigr| \;\in\; [0, 1],
  \label{eq:tld-objective}
\end{equation}
yielding $(\gamma^{*}, T^{*}) = \arg\min_{(\gamma,T)} D_{\mathrm{tld}}$. As a robustness check, we also evaluated the cell-level $L_1$ matrix distance $D_{1}(\gamma,T) = \frac{1}{2M}\sum_{ij}|f^{\mathrm{obs}}_{ij} - f^{\mathrm{grav}}_{ij}(\gamma,T)|$ (reported in SM Sec.~S3), adopting $D_{\mathrm{tld}}$ as our primary metric because $(\gamma, T)$ directly govern the travel-time kernel $e^{-t^{\gamma}/T}$ whereas $D_1$ is diluted by sampling noise across sparse peripheral cells of the OD matrix.

\section{Results across the six metropolitan areas}
\label{app:sixcities}

The main text presents Tokyo; this appendix collects the corresponding results for the six metropolitan areas. Figure~\ref{app_fig_arrest} shows the arrest diagram of each city with its calibrated point, Fig.~\ref{app_fig_calibration} the calibration landscape from which that point is obtained, and Fig.~\ref{app_fig_poisson} the independent Poisson validation of the deterrence form. In every city the stretched-exponential deterrence is selected with $0 < \gamma^{*} < 1$, and the calibrated point $(\gamma^{*}, T^{*})$ falls inside the arrested region: the conclusions of the main text are not particular to Tokyo. The $L_1$ (matrix-level) calibration objective and its residual decomposition for all six cities are provided in SM Sec.~S3, and the complete dynamical analysis of Sendai in SM Sec.~S4.

\begin{figure*}[htbp]
\centering
\includegraphics[width=0.8\linewidth]{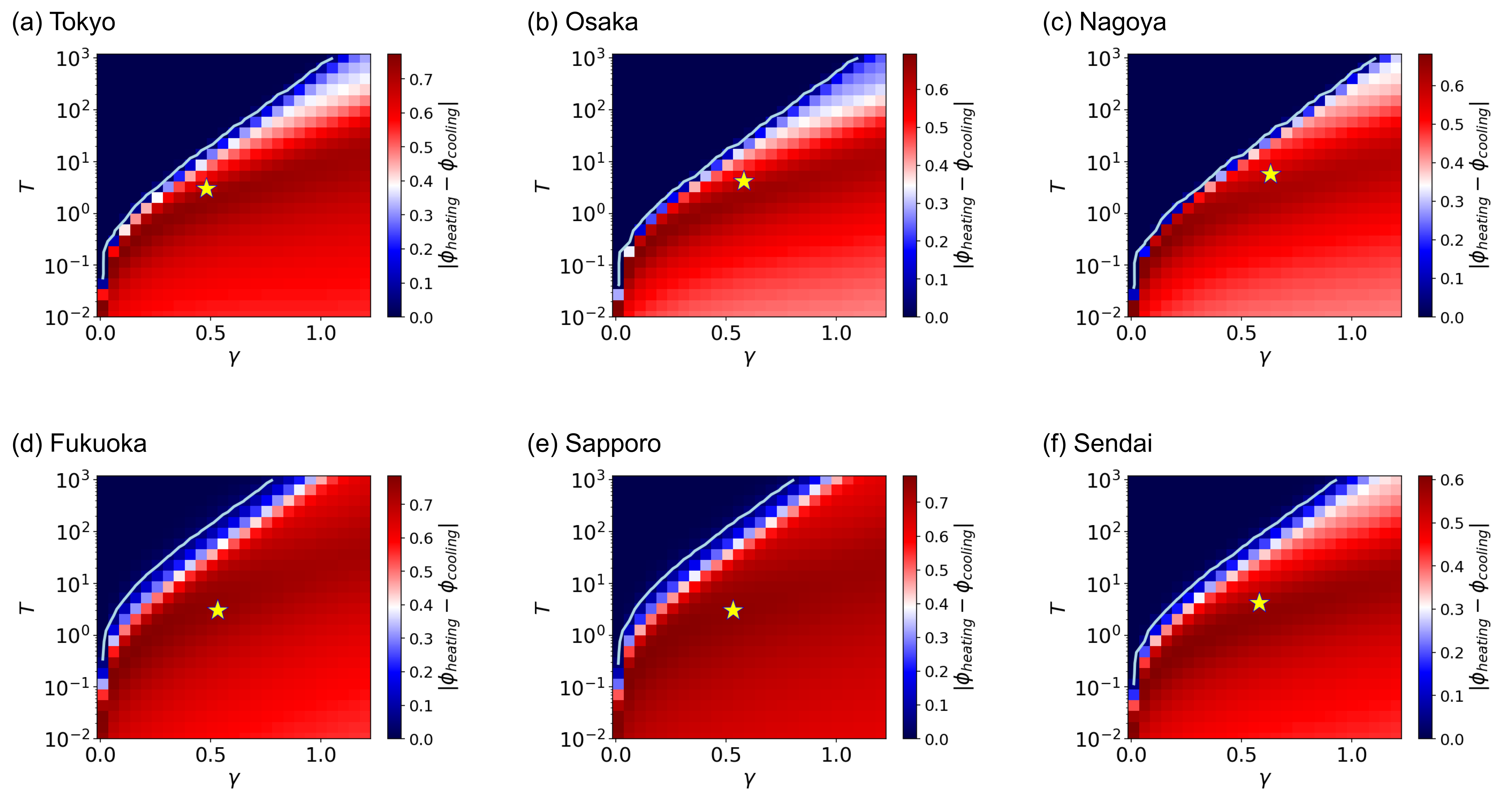}
\caption{\textbf{Every calibrated city lies inside the arrested regime.}
Arrest diagrams of the six metropolitan areas: (a)~Tokyo, (b)~Osaka, (c)~Nagoya, (d)~Fukuoka, (e)~Sapporo, (f)~Sendai.
Each panel shows the protocol-history dependence $|\phi_{\mathrm{heating}} - \phi_{\mathrm{cooling}}|$ over the $(\gamma, T)$ plane under the stepped protocols of Sec.~\ref{metho:subsec_protocols}; the light-blue line marks the ergodic/arrested crossover, and the star marks the calibrated $(\gamma^{*}, T^{*})$ of Fig.~\ref{app_fig_calibration}.}
\label{app_fig_arrest}
\end{figure*}

\begin{figure*}[htbp]
\centering
\includegraphics[width=0.8\linewidth]{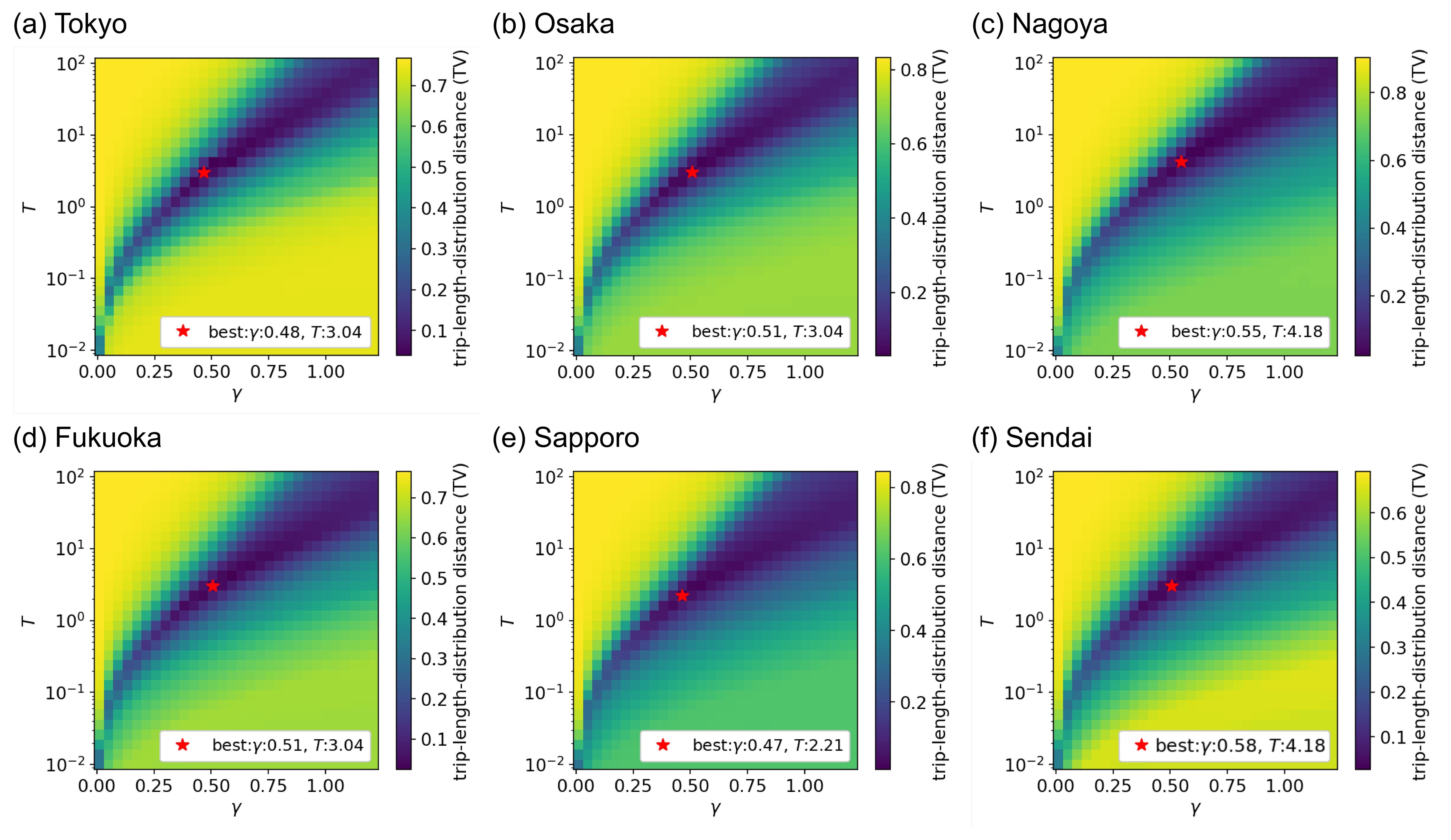}
\caption{%
Trip-length-distribution calibration of the six metropolitan areas: (a)~Tokyo, (b)~Osaka, (c)~Nagoya, (d)~Fukuoka, (e)~Sapporo, (f)~Sendai. Each panel shows the trip-length-distribution objective $D_{\mathrm{tld}}(\gamma, T)$ [Eq.~\eqref{eq:tld-objective}] over the $(\gamma, \log T)$ grid, computed from the Sinkhorn (IPF) thermal equilibrium at each parameter pair; the star marks the minimum.}
\label{app_fig_calibration}
\end{figure*}

\begin{figure*}[t]
\centering
\includegraphics[width=0.8\linewidth]{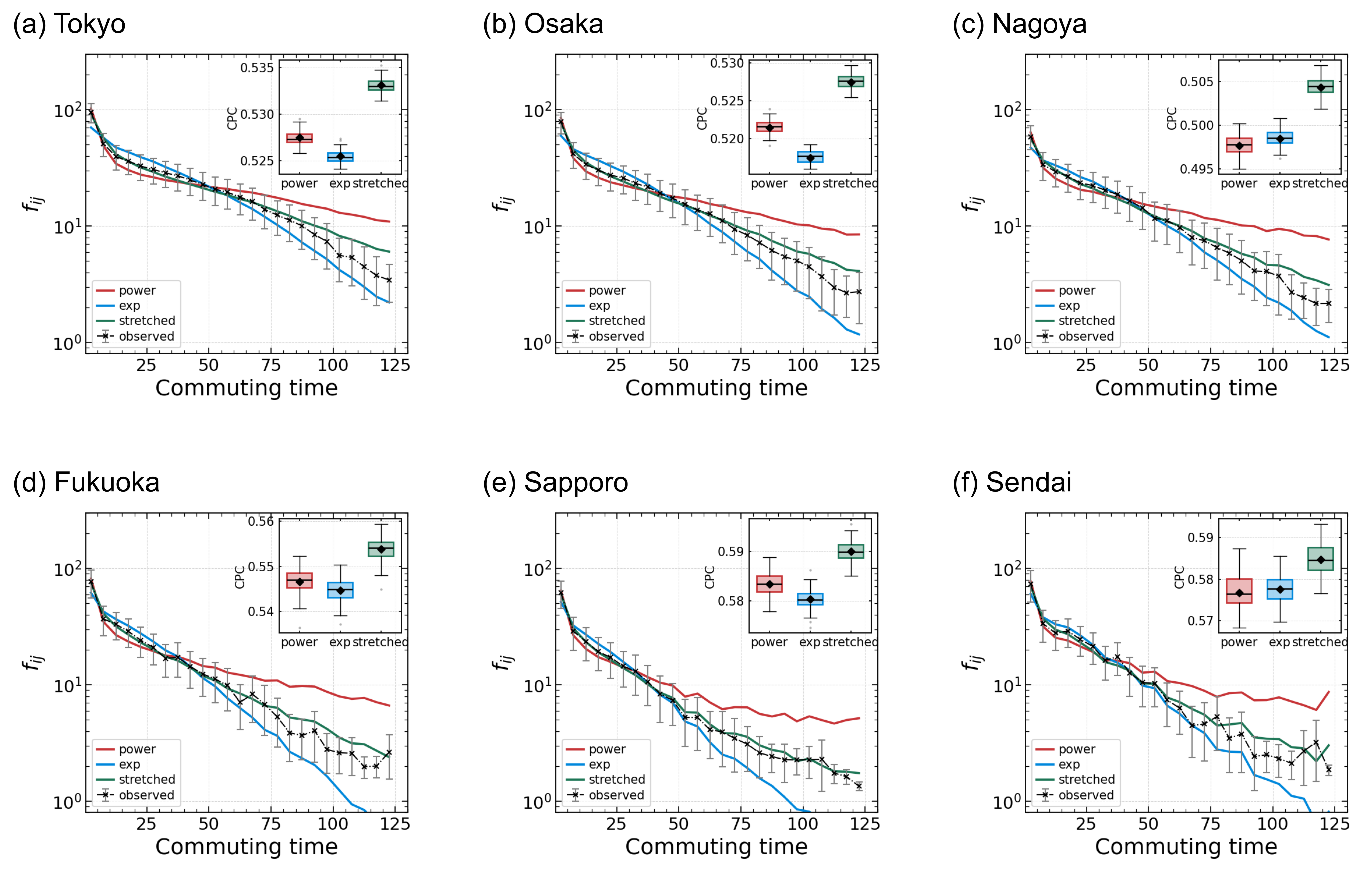}
\caption{\textbf{The stretched exponential attains the highest CPC in every city.}
Independent validation of the deterrence form by Poisson regression (Sec.~\ref{metho:subsubsec_poisson}): (a)~Tokyo, (b)~Osaka, (c)~Nagoya, (d)~Fukuoka, (e)~Sapporo, (f)~Sendai.
Each panel shows flow $f_{ij}$ (log scale) versus commuting time $t_{ij}$ (minutes), with OD pairs grouped into bins.
Markers: observed medians; error bars: interquartile range.
Solid curves: binned medians of model flows $f_{ij}^{*}$ under the three deterrence families.
The stretched-exponential exponents are fitted to $\gamma = 0.4$--$0.5$; this form tracks the data over the full range, while the power law overshoots the tail and the exponential undershoots intermediate times.
Insets: CPC [Eq.~\eqref{eq:cpc}] under Poisson resampling; box spans IQR, whiskers $1.5\,\mathrm{IQR}$, diamond: mean.
The fitted exponents are consistent with $\gamma^{*}$ of Fig.~\ref{app_fig_calibration}.}
\label{app_fig_poisson}
\end{figure*}

\begin{table}[htbp]
 \caption{\label{app_tab:pois_fit}Comparison of the
  stretched-exponential parameters $(\gamma, T)$ estimated by the two
  independent procedures: the grid search minimizing the
  trip-length-distribution distance (TV) from the trip-length-distribution calibration, and the maximum-likelihood Poisson regression of
  the deterrence function $f(t)\propto\exp(-t^{\gamma}/T)$
  (Sec.~\ref{metho:subsubsec_poisson}, Fig.~\ref{app_fig_poisson}).}
  \begin{ruledtabular}
  \begin{tabular}{lcccc}
  & \multicolumn{2}{c}{Grid search (TLD)} &
    \multicolumn{2}{c}{Poisson regression} \\
  \cline{2-3}\cline{4-5}
  Metropolitan area & $\gamma^{*}$ & $T^{*}$ & $\gamma$ & $T$ \\
  \hline
  Tokyo   & 0.48 & 3.04 & 0.40 & 2.96 \\
  Osaka   & 0.58 & 4.18 & 0.40 & 2.84 \\
  Nagoya  & 0.63 & 5.74 & 0.50 & 4.21 \\
  Fukuoka & 0.53 & 3.04 & 0.45 & 4.18 \\
  Sapporo & 0.53 & 3.04 & 0.45 & 4.00 \\
  Sendai  & 0.58 & 4.18 & 0.50 & 4.80 \\
  
  \end{tabular}
  \end{ruledtabular}
\end{table}

\begin{table}[htbp]
  \caption{\label{app_tab:model_comparison}Comparison of the three deterrence
  functions across metropolitan areas. For each functional form we report
  the common part of commuters (CPC) and the AIC difference
  $\Delta\mathrm{AIC} = \mathrm{AIC} - \mathrm{AIC}_{\mathrm{str}}$
  relative to the stretched exponential, which attains the lowest AIC
  ($\Delta\mathrm{AIC} = 0$) in all six areas.}
  \begin{ruledtabular}
  \begin{tabular}{lccccc}
  & \multicolumn{3}{c}{CPC} & \multicolumn{2}{c}{$\Delta$AIC} \\
  Metropolitan area & Power & Exp. & Str.\ exp. & Power & Exp. \\
  \hline
  Tokyo   & 0.5275 & 0.5255 & 0.5332 & 366{,}371 & 597{,}210 \\
  Osaka   & 0.5215 & 0.5175 & 0.5275 & 198{,}232 & 308{,}843 \\
  Nagoya  & 0.4976 & 0.4985 & 0.5043 & 156{,}571 & 89{,}286 \\
  Fukuoka & 0.5467 & 0.5447 & 0.5538 & 55{,}230 & 54{,}727 \\
  Sapporo & 0.5833 & 0.5803 & 0.5900 & 44{,}211 & 48{,}266 \\
  Sendai  & 0.5768 & 0.5776 & 0.5847 & 22{,}024 & 17{,}653 \\
  \end{tabular}
  \end{ruledtabular}
\end{table}

\section{Saddle-point derivation of the concentration results}
\label{app:equiv}

This appendix derives the concentration statements used in Sec.~\ref{result:subsec_equivalence}. Grouping configurations by cost density $e = E[\mathbf{f}]/M$ and using the entropy density \eqref{eq:entropy-density},

\begin{equation}
  Z = \sum_{\mathbf{f} \in \Gamma} \Omega(\mathbf{f})
      e^{-\beta E[\mathbf{f}]}
    \;\asymp\; \int de\; e^{M[s(e) - \beta e]} .
  \label{eq:Zlaplace}
\end{equation}
As $M \to \infty$ the integral is dominated by the maximizer $e^{*}$ of the exponent, with stationarity condition
\begin{equation}
  s'(e^{*}) = \beta ,
  \label{eq:saddle}
\end{equation}
which is unique by Proposition~\ref{prop:concave}, so $-M^{-1}\ln Z \to \min_e[\beta e - s(e)]$, the Legendre transform of $s$. Expanding the exponent to second order about $e^{*}$ gives Gaussian cost fluctuations with variance $-1/[M s''(e^{*})]$, hence the $M^{-1/2}$ relative fluctuation of Eq.~\eqref{eq:fluctuations}.

A caution for the numerical construction of $s(e)$: the parametric sweep---solving the gravity problem at each $\beta$ and recording $(e(\beta), s(\beta))$---reconstructs the concave envelope of $s$ by construction, since a supporting line cannot touch the interior of a nonconcave segment. The sweep is a valid tracing of $s(e)$ here only because Proposition~\ref{prop:concave} establishes concavity independently.


\end{document}